\documentclass[12 pt]{article}
\usepackage{lmodern} 
\usepackage[T1]{fontenc}
\usepackage{textcomp}
\usepackage[scr=boondoxo,frak=boondox,bb=boondox]{mathalfa}
\usepackage{amsmath}
\usepackage{amsmath}
\allowdisplaybreaks
\usepackage{amssymb}
\usepackage{placeins}
\usepackage{amsthm}
\usepackage{ragged2e}
\usepackage{breqn}
\usepackage{graphicx}
\usepackage{epstopdf}
\usepackage{enumitem}
\usepackage{lipsum}
\usepackage{authblk}
\usepackage[top=3cm, bottom=3cm, left=3cm, right=3cm]{geometry}
\usepackage{fancyhdr}
\usepackage{csvsimple}
\usepackage{pgfplotstable}
\usepackage{appendix}
\usepackage{hyperref}
\usepackage{setspace}
\usepackage{subcaption}
\usepackage{float}
\usepackage{verbatim}
\usepackage{booktabs}
\renewenvironment{abstract}{
\hfill\begin{minipage}{0.95\textwidth}
\rule{\textwidth}{1pt}}
{\par\noindent\rule{\textwidth}{1pt}\end{minipage}}

\makeatletter
\renewcommand\@maketitle{
\hfill
\begin{minipage}{0.95\textwidth}
\vskip 2em
\let\footnote\thanks 
{\LARGE \@title \par }
\vskip 1.5em
{\large \@author \par}
\end{minipage}
\vskip 1em \par
}
\makeatother

\newtheorem{theorem}{Theorem}
\newtheorem{lemma}{Lemma}
\title{\textbf{{\sc Asymptotic of Approximate Least Squares Estimators of 
Parameters Two-Dimensional Chirp Signal}}}
\author[$\dagger$]{Rhythm Grover}
\author[$\dagger$,$\ddagger$]{Debasis Kundu}
\author[$\dagger$]{Amit Mitra}
\affil[$\dagger$]{Department of Mathematics, Indian Institute of Technology Kanpur, Kanpur - 208016, India}
\affil[$\ddagger$]{Corresponding author. Email: kundu@iitk.ac.in}
\date{}
\begin{document}
\maketitle

\begin{abstract}
In this paper, we address the problem of parameter estimation of a 2-D chirp model under the assumption that the errors are stationary. We extend the 2-D periodogram method for the sinusoidal model, to find initial values to use in any iterative procedure to compute the least squares estimators (LSEs) of the unknown parameters, to the 2-D chirp model. Next we propose an estimator, known as the approximate least squares estimator (ALSE), that is obtained by maximising a periodogram-type function and is observed to be asymptotically equivalent to the LSE. Moreover the asymptotic properties of these estimators are obtained under slightly mild conditions than those required for the LSEs. For the multiple component 2-D chirp model, we propose a sequential method of estimation of the ALSEs, that significantly reduces the computational difficulty involved in reckoning the LSEs and the ALSEs.  We perform some simulation studies to see how the proposed method works and a data set has been analysed for illustrative purposes.
\end{abstract}


\section{Introduction}\label{sec:1}
\vspace{-5mm}
A two dimensional chirp signal model is expressed mathematically as follows:
\begin{equation}\begin{split}\label{eq:model_1}
y(m, n) & =  A^0\cos(\alpha^0 m + \beta^0 m^2 + \gamma^0 n + \delta^0 n^2) + B^0\sin(\alpha^0 m + \beta^0 m^2 + \gamma^0 n + \delta^0 n^2) \\
& \quad \quad  + X(m,n);    \quad \quad  m = 1, \cdots, M; \ n = 1, \cdots, N.   \\
\end{split}\end{equation}
\justify
\vspace{-5mm}
Here $y(m, n)$s are the signal observations, $A^0$, $B^0$ are real valued, non-zero \emph{amplitudes} and $\big\{\alpha^0, \gamma^0 \big\}$ and $\big\{\beta^0, \delta^0\big\}$ are the \emph{frequencies} and the \emph{frequency rates}, respectively. The random variables $\{X(m, n)\}$ is a sequence of stationary errors. The explicit assumptions on the error structure are provided in section~\ref{sec:2and3 }.

The above model has been considered in many areas of image processing, particularly in modeling gray images. Several estimation techniques for the unknown parameters of this model have been considered by different authors, for instance, Friedlander and Francos \cite{1996}, Francos and Friedlander \cite{1998}, \cite{1999}, Lahiri \cite{2013_1}, \cite{2013_2} and the references cited therein.

Our goal is to estimate the unknown parameters of the above model, primarily the non-linear parameters, the frequencies $\alpha^0$, $\gamma^0$ and the frequency rates $\beta^0$, $\delta^0$, under certain suitable assumptions. One of the straightforward and efficient ways to do so is to use the least squares estimation method. But since the least squares surface is highly non-linear and iterative methods must be employed for their computation, for these methods to work, we need good starting points for the unknown parameters.

One of the fundamental models in statistical signal processing literature, among the 2-D models, is the 2-D sinusoidal model. This model has different applications in many fields such as biomedical spectral analysis, geophysical perception etc. For references, see Barbieri and Barone \cite{1992}, Cabrera and Bose \cite{1993}, Hua \cite{1992_2} , Zhang and Mandrekar \cite{2001}, Prasad et al, \cite{2011_2}, Nandi et al, \cite{2010} and Kundu and Nandi \cite{2003}.
 
\noindent A 2-D sinusoidal model has the following mathematical expression:
\begin{flalign*}
&y(m, n) =  A^0\cos(m \lambda^0 + n \mu^0) + B^0\sin(m \lambda^0 + n \mu^0) + X(m,n)& \\
 &\quad \quad  \quad \quad  \quad \quad  \quad \quad \quad \quad  \quad \quad  \quad \quad  \quad \quad  \quad \quad  \quad \quad  \quad \quad  \quad \quad  m = 1, \cdots, M; \ n = 1, \cdots, N.& 
\end{flalign*}
For this model as well, the least squares surface is highly non-linear and thus we need good initial values, for any iterative procedure to work. One of the most prevalent methods to find the initial guesses for the 2-D sinusoidal model are the periodogram estimators. These are obtained by maximizing a 2-D periodogram function, which is defined as follows:
\begin{equation*}
I(\lambda,\mu) = \frac{2}{MN}\bigg|\sum_{m = 1}^{M}\sum_{n= 1}^{N} y(m,n)e^{-i(m \lambda + n \mu)}\bigg|^2
\end{equation*}
This periodogram function is maximized over 2-D Fourier frequencies, that is, at  \(\displaystyle\bigg(\frac{\pi k}{M},\frac{\pi j}{N}\bigg)\), for \(k = 1, \cdots, M-1\), and \(j = 1, \cdots, N-1\). The estimators that are obtained by maximising the above periodogram function with respect to $\lambda$ and $\mu$ simultaneously over the continuous space $(0,\pi) \times (0,\pi)$, are known as the approximate least squares estimators (ALSEs). Kundu and Nandi, \cite{2003} proved that the ALSEs are consistent and asymptotically equivalent to the least squares estimators (LSEs).

Analogously, we define a periodogram-type function for the 2-D chirp model defined in equation~(\ref{eq:model_1}), as follows:
\begin{equation}\label{eq:per}
 I(\alpha, \beta, \gamma, \delta) = \frac{2}{M N} \bigg|\sum_{m = 1}^{M}\sum_{n= 1}^{N} y(m,n)e^{-i( \alpha m + \beta m^2 + \gamma n + \delta n^2)}\bigg|^2. 
\end{equation}
To find the initial values, we propose to maximise the above function at the grid points \(\displaystyle\bigg(\frac{\pi k_1}{M}, \frac{\pi k_2}{M^2}, \frac{\pi j_1}{N}, \frac{\pi j_2}{N^2}\bigg)\), \(k_1 = 1, \cdots, M-1\), \(k_2 = 1, \cdots, M^2-1\), \(j_1 = 1, \cdots, N-1\), and \(j_2 = 1  \cdots, N^2-1\), corresponding to the Fourier frequencies of the 2-D sinusoidal model. These starting values can be used in any iterative procedure, to compute the LSEs and ALSEs.

Next we propose to estimate the unknown parameters of model~(\ref{eq:model_1}) by approximate least squares estimation method. In this method, we maximize the periodogram-like function $I(\alpha, \beta, \gamma, \delta)$ defined above, with respect to $\alpha$, $\beta$, $\gamma$ and $\delta$ simultaneously, over $(0,\pi) \times (0,\pi) \times (0, \pi) \times (0,\pi)$. The details on the methodology to obtain the ALSEs are further explained in section~\ref{sec:4}. We prove that these estimators are strongly consistent and asymptotically normally distributed under the assumptions, that are slightly mild than those required for the LSEs. Also, the convergence rates of the ALSEs are same as those of the LSEs. 

The rest of the paper is organized as follows. In the next section we state the model assumptions, some notations and some preliminary results required. In section~\ref{sec:4}, we give a brief description of the methodology. In section~\ref{sec:5}, we study the asymptotic properties of one component 2-D chirp model and in section~\ref{sec:6}, we propose a sequential method to obtain the LSEs and ALSEs for the multicomponent 2-D chirp model and study their asymptotic properties. Numerical experiments and  a simulated data analysis are illustrated in sections~\ref{sec:7} and~\ref{sec:8}. In section~\ref{sec:9}, we conclude the paper. All the proofs are provided in the appendices.

\section{Model Assumptions, Notations and Preliminary Results}\label{sec:2and3 }
\vspace{-5mm}
\underline{\textbf{Assumption 1.}}\label{assump:1} The error $X(m, n)$ is stationary with the following form:
\begin{equation*}\label{eq:error}
X(m, n) = \sum_{j = -\infty}^{\infty}\sum_{k = -\infty}^{\infty} a(j, k)\epsilon(m -j, n - k),
\end{equation*}
where $\{\epsilon(m, n)\}$ is a double array sequence of independently and identically distributed (i.i.d.) random variables with mean zero, variance $\sigma^2$ and finite fourth moment, and $a(j, k)$s are real constants such that 
\begin{equation*}\label{eq:constants_err}
 \sum_{j = -\infty}^{\infty}\sum_{k = -\infty}^{\infty}|a(j, k)| < \infty.
\end{equation*} 
\justify
We will use the following notation: $\boldsymbol{\theta} = (A, B, \alpha, \beta, \gamma,\delta)$, the parameter vector, $\boldsymbol{\theta}^0 = (A^0, B^0, \alpha^0, \beta^0, \gamma^0,\delta^0)$, the true parameter vector, $\Theta = (-\infty, \infty) \times  (-\infty, \infty) \times (0, \pi) \times (0, \pi) \times (0, \pi) \times (0, \pi)$, the parameter space.
Also, $\boldsymbol{\vartheta} = (\alpha, \beta, \gamma,\delta)$, a vector of the non-linear parameters.
\vspace{-5mm} 
\justify
\underline{\textbf{Assumption 2.}}\label{assump:2} The true parameter vector $\boldsymbol{\theta}^0$ is an interior point of $\Theta$. \\ \\
Note that the assumptions required to prove strong consistency of the LSEs of the unknown parameters in this case are slightly different from those required to prove the consistency of ALSEs. For the LSEs the parametric space for the linear parameters has to be bounded, though here we do not require that bound. For details on the assumptions for the consistency of the LSEs, see 
Lahiri \cite{2013_1}. \\ \\ 
We need the following results to proceed further:
\begin{lemma}\label{lemma:1}
If $(\omega_1, \omega_2, \psi_1, \psi_2) \in (0, \pi) \times (0, \pi) \times (0, \pi) \times (0, \pi)$, then except for a countable number of points, and for s, t = 0, 1, $\cdots$, the following are true:
\vspace{-4mm}
\begin{enumerate}[label=(\alph*)]
\item \hspace{-4mm} $\lim\limits_{\textnormal{min}{\{M, N\} \rightarrow} \infty} \frac{1}{M N}  \sum\limits_{n= 1}^{N}\sum\limits_{m = 1}^{M} \cos(\omega m^2 + \psi n^2) = \hspace{-4mm}\lim\limits_{\textnormal{min}{\{M, N\} \rightarrow} \infty} \frac{1}{M N}  \sum\limits_{n= 1}^{N}\sum\limits_{m = 1}^{M} \sin(\omega m^2 + \psi n^2) = 0$ \\
\item \hspace{-6mm} $\lim\limits_{\textnormal{min}{\{M, N\} \rightarrow} \infty} \frac{1}{M N}  \sum\limits_{n=1}^{N}\sum\limits_{m = 1}^{M} \cos^2(\omega m^2 + \psi n^2) = \hspace{-4mm}\lim\limits_{\textnormal{min}{\{M, N\} \rightarrow} \infty} \frac{1}{M N}  \sum\limits_{n= 1}^{N}\sum\limits_{m = 1}^{M} \sin^2(\omega m^2 + \psi n^2) = \frac{1}{2}$ \\
\item $\lim\limits_{\textnormal{min}{\{M, N\} \rightarrow} \infty} \frac{1}{M N}  \sum\limits_{n= 1}^{N}\sum\limits_{m = 1}^{M} \cos(\omega_1 m + \omega_2 m^2   + \psi_1 n + \psi_2 n^2) =  0,$ \\
\item $\lim\limits_{\textnormal{min}{\{M, N\} \rightarrow} \infty} \frac{1}{M N}  \sum\limits_{n= 1}^{N}\sum\limits_{m = 1}^{M} \sin(\omega_1 m + \omega_2 m^2   + \psi_1 n + \psi_2 n^2) = 0,$ \\
\item $\lim\limits_{\textnormal{min}{\{M, N\} \rightarrow} \infty} \frac{1}{M^{(s+1)} N^{(t+1)}}  \sum\limits_{n= 1}^{N}\sum\limits_{m = 1}^{M} m^s n^t\cos^2(\omega_1 m + \omega_2 m^2   + \psi_1 n + \psi_2 n^2) =  \frac{1}{2(s+1)(t+1)},$ \\
\item $\lim\limits_{\textnormal{min}{\{M, N\} \rightarrow} \infty} \frac{1}{M^{(s+1)} N^{(t+1)}}  \sum\limits_{n= 1}^{N}\sum\limits_{m = 1}^{M} m^s n^t\sin^2(\omega_1 m + \omega_2 m^2   + \psi_1 n + \psi_2 n^2) =  \frac{1}{2(s+1)(t+1)},$  \\
\item $\lim\limits_{\textnormal{min}{\{M, N\} \rightarrow} \infty} \sup\limits_{\alpha, \beta, \gamma, \delta} |\frac{1}{M^{(s+1)} N^{(t+1)}} \sum\limits_{n= 1}^{N}\sum\limits_{m = 1}^{M} m^s n^t X(m, n)e^{i(\alpha m + \beta m^2 + \gamma n + \delta n^2)}|$ $\rightarrow$ $0$ $a.s. $
\end{enumerate}
\end{lemma}
\begin{proof}
Refer to Lahiri \cite{2013_1} \\
\end{proof}
\begin{lemma}\label{lemma:4}
If $(\omega , \psi)$ $\in$ $(0,\pi) \times (0,\pi)$, then except for a countable number of points, the following holds true:
$$\lim\limits_{n \rightarrow \infty} \frac{1}{n^k\sqrt{n}} \sum_{t=1}^{n} t^k \cos(\omega t + \psi t^2) = \lim\limits_{n \rightarrow \infty} \frac{1}{n^k\sqrt{n}} \sum_{t=1}^{n} t^k \sin(\omega t + \psi t^2) = 0;\ k = 0, 1, 2, \cdots $$
\end{lemma}
\begin{proof}
Refer to Lahiri \cite{2013_1}.\\
\end{proof}
\begin{lemma}\label{lemma:5}
If $(\omega_1, \omega_2, \omega_3, \omega_4) \in (0, \pi) \times (0, \pi) \times (0, \pi) \times (0, \pi)$ and $(\psi_1, \psi_2, \psi_3, \psi_4) \in (0, \pi) \times (0, \pi) \times (0, \pi) \times (0, \pi)$, then except for a countable number of points, and for s, t = 0, 1, $\cdots$, the following are true:
\begin{enumerate}[label=(\alph*)]
\item
\begingroup
\allowdisplaybreaks
\parbox[t]{\textwidth}{
   \vspace{-3em} \begin{align*}
& \lim\limits_{\textnormal{min}{\{M, N\} \rightarrow} \infty} \frac{1}{M^s N^t\sqrt{M N}}  \sum\limits_{n= 1}^{N}\sum\limits_{m = 1}^{M} m^s n^t \cos(\omega_1 m + \omega_2 m^2   + \omega_3 n + \omega_4 n^2)\times \\ 
& \qquad \qquad \qquad \qquad \qquad \qquad \qquad \qquad \qquad  \cos(\psi_1 m + \psi_2 m^2   + \psi_3 n + \psi_4 n^2) =  0,
\end{align*}}%
\endgroup
\item
\begingroup
\allowdisplaybreaks
\parbox[t]{\textwidth}{
   \vspace{-3em} \begin{align*}
&\lim\limits_{\textnormal{min}{\{M, N\} \rightarrow} \infty} \frac{1}{M^s N^t\sqrt{M N}}  \sum\limits_{n= 1}^{N}\sum\limits_{m = 1}^{M} m^s n^t \sin(\omega_1 m + \omega_2 m^2 + \omega_3 n + \omega_4 n^2) \times \\
&\qquad \qquad \qquad \qquad \qquad \qquad \qquad \qquad \qquad \sin(\psi_1 m + \psi_2 m^2 + \psi_3 n + \psi_4 n^2) = 0,
\end{align*}}%
\endgroup
\item
\begingroup
\allowdisplaybreaks
\parbox[t]{\textwidth}{
   \vspace{-3em} \begin{align*}
&\lim\limits_{\textnormal{min}{\{M, N\} \rightarrow} \infty} \frac{1}{M^s N^t\sqrt{M N}}  \sum\limits_{n= 1}^{N}\sum\limits_{m = 1}^{M} m^s n^t \sin(\omega_1 m + \omega_2 m^2 + \omega_3 n + \omega_4 n^2)\times \\
&\qquad \qquad \qquad \qquad \qquad \qquad \qquad \qquad \qquad \cos(\psi_1 m + \psi_2 m^2 + \psi_3 n + \psi_4 n^2) = 0.
\end{align*}}%
\endgroup
\end{enumerate}
\end{lemma}
\begin{proof}
See~\nameref{appendix:D}.\\ 
\end{proof}
\vspace{-10mm}
\section{Method to obtain ALSEs}\label{sec:4}
\vspace{-4mm}
Consider the periodogram-like function defined in~(\ref{eq:per}). In matrix notation, it can be written as:
\vspace{-2mm}
\begin{equation*}\label{eq:per_matrix}
I(\boldsymbol{\vartheta}) =  \frac{2}{M N} Y^{T} W(\boldsymbol{\vartheta})W(\boldsymbol{\vartheta})^{T}Y.
\end{equation*}
\vspace{-12mm}
\justify
Here, $Y_{M N \times 1} =  \left[\begin{array}{ccccccc}y(1, 1) & \cdots & y(M, 1) & \cdots &  y(1, N) & \cdots & y(M, N)\end{array}\right]^{T}$ is the observed data vector, and
\begin{equation*}
W(\boldsymbol{\vartheta})_{M N \times 2} = \left[\begin{array}{cc}\cos(\alpha + \beta + \gamma + \delta) & \sin(\alpha + \beta + \gamma + \delta) \\ \cos(2\alpha + 4\beta + \gamma + \delta) & \sin(2\alpha + 4\beta + \gamma + \delta) \\ \vdots & \vdots \\ \cos(M \alpha + M^2 \beta + \gamma + \delta) &  \sin(M \alpha + M^2 \beta + \gamma + \delta) \\ \vdots & \vdots \\  \cos(\alpha + \beta + N \gamma + N^2 \delta) & \sin(\alpha + \beta + N \gamma + N^2 \delta) \\ \cos(2\alpha + 4\beta + N \gamma + N^2 \delta) & \sin(2\alpha + 4\beta + N \gamma + N^2 \delta)\\\vdots & \vdots \\ \cos(M \alpha + M^2 \beta + N \gamma + N^2 \delta) & \sin(M \alpha + M^2 \beta + N \gamma + N^2 \delta)\end{array}\right]
\end{equation*}
In matrix notation, equation~(\ref{eq:model_1}), can be written as:
\begin{equation*}\label{eq:model_mat}
Y = W(\boldsymbol{\vartheta})\boldsymbol{\phi} + X,
\end{equation*}
\vspace{-16mm}
\justify
where $X_{M N \times 1} = \left[\begin{array}{ccccccc}X(1, 1) & \cdots & X(M, 1) & \cdots & X(1, N) & \cdots & X(M, N)\end{array}\right]^{T}$
is the error vector, and $\boldsymbol{\phi} = \left[\begin{array}{cc}A & B\end{array}\right]^{T}$.
The estimators obtained by maximising the function $I(\boldsymbol{\vartheta})$ are known as the approximate least squares estimators (ALSEs). 
We will show that the estimators obtained by maximising $I(\boldsymbol{\vartheta})$ are asymptotically equivalent to the estimators obtained by minimising the error sum of squares function, that is the LSEs, and hence the former are termed as the ALSEs. To do so, we require the following lemma:
\begin{lemma}\label{lemma:6}
For $\boldsymbol{\vartheta} \in (0,\pi) \times (0, \pi) \times (0,\pi) \times(0,\pi)$, except for a countable number of points, we have the following result: 
\begin{equation*} 
\frac{1}{M N} W(\boldsymbol{\vartheta})^{T}W(\boldsymbol{\vartheta}) \rightarrow
\left[\begin{array}{cc} 1/2 & 0 \\ 0 & 1/2 \end{array}\right].
\end{equation*}
\end{lemma}
\begin{proof}
Consider the following:
\begin{equation*}
\frac{1}{M N} W(\boldsymbol{\vartheta})^{T}W(\boldsymbol{\vartheta}) = \begin{bmatrix}
\Omega_{11} & \Omega_{12}\\
 \Omega_{21} & \Omega_{22}
\end{bmatrix},
\end{equation*}
where,
\begin{flalign*}
&\Omega_{11} = \sum\limits_{n=1}^{N}\sum\limits_{m=1}^{M}\cos^2(\alpha m + \beta m^2 + \gamma n + \delta n^2),&\\
&\Omega_{12} = \sum\limits_{n=1}^{N}\sum\limits_{m=1}^{M}\cos(\alpha m + \beta m^2 + \gamma n + \delta n^2)\sin(\alpha m + \beta m^2 + \gamma n + \delta n^2),& \\
&\Omega_{21} = \sum\limits_{n=1}^{N}\sum\limits_{m=1}^{M}\cos(\alpha m + \beta m^2 + \gamma n + \delta n^2)\sin(\alpha m + \beta m^2 + \gamma n + \delta n^2),& \\
& \Omega_{22} = \sum\limits_{n=1}^{N}\sum\limits_{m=1}^{M}\sin^2(\alpha m + \beta m^2 + \gamma n + \delta n^2).
\end{flalign*}
Now using Lemma$~\ref{lemma:1}$ \textit{(c)}, \textit{(e)} and \textit{(f)}, it can be easily seen that the matrix on the right hand side of the above equation tends to $\begin{bmatrix}
1/2 & 0 \\
0 & 1/2
\end{bmatrix}$, except for a countable number of points and hence the result.\\
\end{proof}
\justify
\vspace{-5mm}
We know that to find the LSEs, we minimise the following error sum of squares: 
\begin{equation}\label{eq:ess}
Q(\boldsymbol{\theta}) = (Y - W(\boldsymbol{\vartheta}) \boldsymbol{\phi})^{T}(Y - W(\boldsymbol{\vartheta}) \boldsymbol{\phi})
\end{equation}
with respect to $\boldsymbol{\theta}$. If we fix $\boldsymbol{\vartheta}$, then the estimates of the linear parameters can be obtained by separable regression technique of Richards \cite{1961} by minimizing $Q(\boldsymbol{\theta})$ with respect to $A$ and $B$.
Thus the estimate of $\boldsymbol{\phi}^0 =  \left[\begin{array}{cc}A^0 & B^0\end{array}\right]^{T}$ is given by: \\
\begin{equation}\label{eq:lin_est}
\hat{\boldsymbol{\phi}}(\boldsymbol{\vartheta}) = \left[\begin{array}{c}\hat{A}(\boldsymbol{\vartheta}) \\ \hat{B}(\boldsymbol{\vartheta})\end{array}\right] = (W(\boldsymbol{\vartheta})^{T}W(\boldsymbol{\vartheta}))^{-1} W(\boldsymbol{\vartheta})^{T}Y.
\end{equation}
Substituting $\hat{A}(\boldsymbol{\vartheta})$ and $\hat{B}(\boldsymbol{\vartheta})$ in~(\ref{eq:ess}), we have:
\begin{equation*}\label{eq:ess_non_lin}
Q(\hat{A}(\boldsymbol{\vartheta}), \hat{B}(\boldsymbol{\vartheta}),\boldsymbol{\vartheta}) = Y^{T}(I - W(\boldsymbol{\vartheta})(W(\boldsymbol{\vartheta})^{T}W(\boldsymbol{\vartheta}))^{-1} W(\boldsymbol{\vartheta})^{T})Y.
\end{equation*}
Using Lemma~\ref{lemma:6}, we have the following relationship between the function $Q(\boldsymbol{\theta})$ and the periodogram-like function $I(\boldsymbol{\vartheta})$:
$$\frac{1}{MN}Q(\hat{A}(\boldsymbol{\vartheta}), \hat{B}(\boldsymbol{\vartheta}),\boldsymbol{\vartheta}) = \frac{1}{MN}Y^{T}Y - I(\boldsymbol{\vartheta}) + o(1).$$
Here, a function $f$ is $o(1)$, if $f$ $\rightarrow$ 0 as min\{$M$, $N$\} $\rightarrow$ $\infty$ Thus, $\hat{\boldsymbol{\vartheta}}$ that minimises $Q(\hat{A}(\boldsymbol{\vartheta}), \hat{B}(\boldsymbol{\vartheta}),\boldsymbol{\vartheta})$ is equivalent to $\tilde{\boldsymbol{\vartheta}}$, which maximises $I(\boldsymbol{\vartheta})$.
\vspace{-5mm}
\section{Asymptotic Properties of ALSEs}\label{sec:5}
In this section, we study the asymptotic properties of the proposed estimators, the ALSEs of model~(\ref{eq:model_1}). The following theorem states the result on the consistency property of the ALSEs.
\begin{theorem}\label{theorem:1}
If the assumptions 1 and 2 are satisfied, then $\tilde{\boldsymbol{\theta}} = (\tilde{A}, \tilde{B}, \tilde{\alpha}, \tilde{\beta}, \tilde{\gamma}, \tilde{\delta})$, the ALSE of $\boldsymbol{\theta}^0$, is a strongly consistent estimator of $\boldsymbol{\theta}^0$, that is, $\tilde{\boldsymbol{\theta}} \xrightarrow{a.s.} \boldsymbol{\theta}^0$ as $\textmd{min}\{M, N\} \rightarrow \infty$.
\end{theorem}
\begin{proof}
See~\nameref{appendix:A}.\\
\end{proof}
\justify
\vspace{-8mm}
In the following theorem, we state the result obtained on the asymptotic distribution of the proposed estimators.
\begin{theorem}\label{theorem:2}
If the assumptions 1 and 2 are true, then the asymptotic distribution of $(\tilde{\boldsymbol{\theta}} - \boldsymbol{\theta}^0)\mathbf{D}^{-1}$ is same as that of $(\hat{\boldsymbol{\theta}} - \boldsymbol{\theta}^0)\mathbf{D}^{-1}$ as $\textmd{min}\{M, N\} \rightarrow \infty$, where 
 $\tilde{\boldsymbol{\theta}}$ = $(\tilde{A}, \tilde{B}, \tilde{\alpha}, \tilde{\beta},  \tilde{\gamma}, \tilde{\delta})$ is the ALSE of $\boldsymbol{\theta}^0$ and $\hat{\boldsymbol{\theta}}$ = $(\hat{A}, \hat{B}, \hat{\alpha}, \hat{\beta},  \hat{\gamma}, \hat{\delta})$ is the LSE of $\boldsymbol{\theta}^0$ and $\mathbf{D}$ is a 6 $\times$ 6 diagonal matrix defined as:\\
$ \mathbf{D}$ = $\textnormal{diag}(M^{\frac{-1}{2}}N^{\frac{-1}{2}}, M^{\frac{-1}{2}}N^{\frac{-1}{2}}, M^{\frac{-3}{2}}N^{\frac{-1}{2}}, M^{\frac{-5}{2}}N^{\frac{-1}{2}}, M^{\frac{-1}{2}}N^{\frac{-3}{2}}, M^{\frac{-1}{2}}N^{\frac{-5}{2}}).$
\end{theorem}
 \begin{proof}
See~\nameref{appendix:B}.\\
\end{proof}
\section{Multiple Component 2-D Chirp Model}\label{sec:6}
In this section, we consider a 2-D chirp model with multiple components, mathematically expressed in the following form:
\begin{equation}\begin{split}\label{eq:model_mul_comp}
y(m, n) & =  \sum_{k=1}^p \bigg(A_k^0\cos(\alpha_k^0 m + \beta_k^0 m^2 + \gamma_k^0 n + \delta_k^0 n^2) + B_k^0\sin(\alpha_k^0 m + \beta_k^0 m^2 + \gamma_k^0 n + \delta_k^0 n^2)\bigg) \\
& \quad \quad  + X(m,n);    \quad \quad  m = 1, \cdots, M; \ n = 1, \cdots, N.
\end{split}\end{equation}
Here $y(m, n)$ is the observed data vector, $A_k^0$s, $B_k^0$s are the amplitudes, $\alpha_k^0$s, $\gamma_k^0$s are the frequencies and the $\beta_k^0$s, $\delta_k^0$s are the frequency rates. The random variables sequence $\{X(m, n)\}$ is a stationary error sequence. In practice, the number of components, $p$ is unknown and its estimation is an important and still an open problem. For recent references on this model, see Zhang et al. \cite{2008} and Lahiri \cite{2013_1}.  

\noindent Here it is assumed that $p$ is known and our main purpose is to estimate the unknown parameters of this model, primarily the non-linear parameters. Finding the ALSEs for the above model is computationally challenging, especially when the number of components, $p$ is large. Even when $p = 2$, we need to solve a 12-D optimisation problem to obtain the ALSEs. Thus, we propose a sequential procedure to find these estimates. This method reduces the complexity of computation without compromising on the efficiency of the estimators. We prove that the ALSEs obtained by the proposed sequential procedure are strongly consistent and have the same rates of convergence as the LSEs. 

\noindent In the following subsection, we provide the algorithm to obtain the sequential ALSEs of the unknown parameters of the $p$ component 2-D chirp signal. Let us denote $\boldsymbol{\vartheta}_{k} = (\alpha_k, \beta_k, \gamma_k, \delta_k)$.
\vspace{-9mm}
\subsection{Algorithm to find the ALSEs:}
\textit{Step 1:} Maximizing the periodogram-like function
\begin{equation}\begin{split}\label{eq:per_1}
I_{1}(\boldsymbol{\vartheta})& = \frac{1}{MN} \bigg(\sum_{m=1}^{M}\sum_{n=1}^{N}y(m, n)\cos(\alpha m + \beta m^2 + \gamma n + \delta n^2)\bigg)^2 + \\
& \quad \frac{1}{MN} \bigg(\sum_{m=1}^{M}\sum_{n=1}^{N}y(m, n)\sin(\alpha m + \beta m^2 + \gamma n + \delta n^2)\bigg)^2.
\end{split}\end{equation}
We first obtain the non-linear parameter estimates: $\tilde{\boldsymbol{\vartheta}_1} = (\tilde{\alpha_1}, \tilde{\beta_1}, \tilde{\gamma_1}, \tilde{\delta_1})$.
Then the linear parameter estimates can be obtained by substituting  $\tilde{\boldsymbol{\vartheta}_1}$ in~(\ref{eq:lin_est}).
Thus
\begin{equation}\begin{split}\label{eq:lin_ALSEs}
\tilde{A_1} = \frac{2}{M N}\sum\limits_{n = 1}^{N}\sum\limits_{m = 1}^{M} y(m, n) \cos(\tilde{\alpha_1} m + \tilde{\beta_1} m^2 + \tilde{\gamma_1} n + \tilde{\delta_1} n^2),\\ 
\tilde{B_1} = \frac{2}{M N}\sum\limits_{n= 1}^{N}\sum\limits_{m = 1}^{M} y(m, n) \sin(\tilde{\alpha_1} m + \tilde{\beta_1} m^2 + \tilde{\gamma_1} n + \tilde{\delta_1} n^2). \\ 
\end{split}\end{equation}
\textit{Step 2:} Now we have the estimates of the parameters of the first component of the observed signal. We subtract the contribution of the first component from the original signal vector $Y_{MN \times 1}$ to eliminate the effect of the first component and obtain a new data vector, say 
$$Y^{1} = Y - W(\tilde{\boldsymbol{\vartheta}_1})\begin{pmatrix}
\tilde{A_1} \\\tilde{B_1}
\end{pmatrix}.
$$
\textit{Step 3:} Now we compute $\tilde{\boldsymbol{\vartheta}_2} = (\tilde{\alpha_{2}},\tilde{\beta_{2}}, \tilde{\gamma_{2}}, \tilde{\delta_{2}})$ by maximizing $I_{2}(\boldsymbol{\vartheta})$ which is obtained by replacing the original data vector by the new data vector in~(\ref{eq:per_1}) and then the linear parameters, $\tilde{A_2}$ and $\tilde{B_2}$ can be obtained by substituting $\tilde{\boldsymbol{\vartheta}_2}$ in~(\ref{eq:lin_est}). \\ \\
Step 4: Continue the process upto $p$-steps. 
\justify
\vspace{-3mm}
\subsection{Asymptotic Properties}
Further assumptions required to study the consistency property and derive the asymptotic distribution of the proposed estimators, are stated as follows:
\justify
\vspace{-2mm}
\underline{\textbf{Assumption 3.}}\label{assump:3} $\boldsymbol{\theta}_{k}^{0}$ is an interior point of $\Theta$, for all $k = 1, \ldots, p$ and the frequencies $\alpha_{k}^0s$, $\gamma_{k}^0s$ and the frequency rates $\beta_{k}^0s$, $\delta_{k}^0s$ are such that $(\alpha_i^0, \beta_i^0, \gamma_i^0, \delta_i^0) \neq (\alpha_j^0, \beta_j^0, \gamma_j^0, \delta_j^0)$ $\forall i \neq j$. \\ \\
\underline{\textbf{Assumption 4.}}\label{assump:4} $A_k^0$s and $B_k^0$s satisfy the following relationship:
\begin{equation*}
\infty > {A_{1}^{0}}^2 + {B_{1}^{0}}^2 > {A_{2}^{0}}^2 + {B_{2}^{0}}^2 > \cdots > {A_{p}^{0}}^2 + {B_{p}^{0}}^2 > 0.
\end{equation*}
In the following theorems, we state the results we obtained on the consistency of the proposed estimators.
\begin{theorem}\label{theorem:3}
Under the assumptions 1, 3 and 4, $\tilde{A_1}, \tilde{B_1}, \tilde{\alpha_1}$, $\tilde{\beta_1}$, $\tilde{\gamma_1}$ and $\tilde{\delta_1}$ are strongly consistent estimators of $A_{1}^{0}, B_{1}^{0},  \alpha_{1}^{0}$, $\beta_{1}^{0}$, $\gamma_{1}^{0}$, $\delta_{1}^{0}$ respectively, that is, $\tilde{\boldsymbol{\theta}_1} \xrightarrow{a.s.} \boldsymbol{\theta}_1^0$ as $\textmd{min}\{M, N\} \rightarrow \infty$.
\end{theorem}  
\begin{proof}
See~\nameref{appendix:C}.\\
\end{proof}
\begin{theorem}\label{theorem:4}
If the assumptions 1, 3 and 4 are satisfied and p $\geqslant$ 2,then $\tilde{\boldsymbol{\theta}_2} \xrightarrow{a.s.} \boldsymbol{\theta}_2^0$ as $\textmd{min}\{M, N\} \rightarrow \infty$.
\end{theorem}
\begin{proof}
See \nameref{appendix:C}.\\
\end{proof}
\justify
\vspace{-6mm}
The result obtained in the above theorem can be extended upto the $p$-th step. Thus for any $k$ $\leqslant$ $p$, the ALSEs obtained at the $k$-th step are strongly consistent.
\begin{theorem}\label{theorem:5}
If the assumptions 1, 3 and 4 are satisfied, and if $\tilde{A_k}$,  $\tilde{B_k}$, $\tilde{\alpha_k}$, $\tilde{\beta_k}$, $\tilde{\gamma_k}$ and $\tilde{\delta_k}$ are the estimators obtained at the $k$-th step, and k $>$ p then $\tilde{A_k}$ $\xrightarrow{a.s}$ 0 and $\tilde{B_k}$ $\xrightarrow{a.s}$ 0 as $\textmd{min}\{M, N\} \rightarrow \infty$.
\end{theorem}
\begin{proof}
See \nameref{appendix:C}.\\
\end{proof}
\justify
\vspace{-4mm}
Next we derive the asymptotic distribution of the proposed estimators. In the following theorem, we state the results on the distribution of the sequential ALSEs.

\begin{theorem}\label{theorem:6}
If the assumptions, 1, 3 and 4 are satisfied, then 
\begin{equation*}
(\tilde{\boldsymbol{\theta}_1} - \boldsymbol{\theta}_1^0)\mathbf{D}^{-1} \xrightarrow{d} N_6(0, \sigma^2 c \boldsymbol{\Sigma}_1^{-1}) 
\end{equation*}
where $\mathbf{D}$ is the diagonal matrix as defined in Theorem~\ref{theorem:2} and  $c = \sum\limits_{j=-\infty}^{\infty}\sum\limits_{k=-\infty}^{\infty}a(j, k)^2$
$$\boldsymbol{\Sigma}_1^{-1} = \frac{2}{{A_1^0}^2 + {B_1^0}^2}\begin{bmatrix}
{A_1^0}^2 + 17 {B_1^0}^2 & -16 A_1^0B_1^0 & -36 B_1^0 & 30 B_1^0 & -36 B_1^0 & 30 B_1^0 \\
-16 A_1^0B_1^0 & 17 {A_1^0}^2 + {B_1^0}^2 & 36 A_1^0 & -30A_1^0 &  36 A_1^0 & -30A_1^0 \\
-36B_1^0 & 36A_1^0 & 192 & -180 & 0 & 0 \\
30B_1^0 & -30A_1^0 & -180 & 180 & 0 & 0 \\
-36B_1^0 & 36A_1^0 & 0 & 0 & 192 & -180 \\
30B_1^0 & -30A_1^0 & 0 & 0 & -180 & 180 \\ 
\end{bmatrix}$$
\end{theorem}
\begin{proof}
See \nameref{appendix:D}.
\end{proof}
\justify
\vspace{-5mm}
The above result holds true for all $1 \leqslant k \leqslant p$ and is stated in the following theorem.
\begin{theorem}\label{theorem:7}
If the assumptions, 1, 3 and 4 are satisfied, then 
\begin{equation*}
(\tilde{\boldsymbol{\theta}_k} - \boldsymbol{\theta}_k^0)\mathbf{D}^{-1} \xrightarrow{d} N_6(0, \sigma^2 c \boldsymbol{\Sigma}_k^{-1}),
\end{equation*}
where $\boldsymbol{\Sigma}_k^{-1}$ can be obtained by replacing $A_1^0$ by $A_k^0$ and $B_1^0$ by $B_k^0$ in $\boldsymbol{\Sigma}_1^{-1}$ defined above. 
\end{theorem}
\begin{proof}
This proof can be obtained by proceeding exactly in the same manner as in the proof of Theorem~\ref{theorem:6}.\\
\end{proof}
\vspace{-10mm}
\section{Simulation Studies}\label{sec:7}
\vspace{-5mm}
\subsection{Simulation results for the one component model}
We perform numerical simulations on model~(\ref{eq:model_1}) with the following parameters:
$$A^0 = 2,\ B^0 = 3,\ \alpha^0 = 1.5,\ \beta^0 = 0.5,\ \gamma^0 = 2.5\ \textmd{and}\ \delta^0 = 0.75.$$
The following error structures are used to generate the data:
\begin{align}
& 1. \quad X(m, n) = \epsilon(m, n).\label{eq:error_normal}\\
& 2. \quad X(m, n) = \epsilon(m, n) + 0.5  \epsilon(m, n-1) + 0.4 \epsilon(m-1, n) + 0.3 \epsilon(m-1, n-1).\label{eq:error_MA}
\end{align}
\justify
Here $\epsilon(m, n)$ $\sim$ $N(0, \sigma^2)$. For simulations we consider different values of $\sigma$ and different values of $M$ and $N$ as can be seen in the tables.
We estimate the parameters both by least  squares estimation method and approximate least squares estimation method. These estimates are obtained 1000 times each and averages, biases and MSEs are reported. We also compute the asymptotic variances to compare with the corresponding MSEs.
\noindent
From the tables above, it is observed that as the error variance increases, the MSEs also increase for both the LSEs and the ALSEs. As the sample size increases, one can see that the estimates become closer to the corresponding true values, that is, the biases become small. Also, the MSEs decrease as the sample size, $M$ and $N$ increase, and the order of the MSEs of both the estimators is almost equivalent to the order of the asymptotic variances. Hence, one may conclude that they are well matched. The MSEs of the ALSEs get close to those of LSEs as $M$ and $N$ increase and hence to the theoretical asymptotic variances of the LSEs, showing that they are asymptotically equivalent.
\begin{table}[H]
\centering
\resizebox{0.99\textwidth}{!}{\begin{tabular}{||c|c||c|c|c|c||c|c|c|c||}
\hline
\hline
\multicolumn{2}{||c||}{Parameters} & $\alpha$ & $\beta$ & $\gamma$ & $\delta$ & $\alpha$ & $\beta$ & $\gamma$ & $\delta$ \\
\hline  \multicolumn{2}{||c||}{True values} & 1.5 & 0.5 & 2.5 & 0.75 & 1.5 & 0.5 & 2.5 & 0.75 \\
\hline $\sigma$ & & \multicolumn{4}{c||}{ALSEs} & \multicolumn{4}{c||}{LSEs} \\
\hline 0.1 & Avg & 1.4910 & 0.5005 & 2.5194 & 0.7492 & 1.5000 & 0.4999 & 2.5000 & 0.7499\\
\hline & Bias & -0.0090 & 0.0005 & 0.0194 & -0.0008 & 3.85E-05 & -1.66E-06 & 1.22E-05 & -3.97E-07 \\
\hline  & MSE & 8.21E-05 & 2.83E-07 & 3.79E-04 & 5.62E-07 & 8.29E-07 & 1.13E-09 & 7.18E-07 & 9.90E-10 \\
\hline  & AVar & 7.56E-07 & 1.13E-09 & 7.56E-07 & 1.13E-09 & 7.56E-07 & 1.13E-09 & 7.56E-07 & 1.13E-09\\
\hline
\multicolumn{2}{||c||}{} & \multicolumn{4}{c||}{ALSEs} & \multicolumn{4}{c||}{LSEs} \\
\hline 
0.5 & Avg & 1.4912 & 0.5005 & 2.5196 & 0.7492 & 1.5000 & 0.5000 & 2.5003 & 0.7499\\\hline  & Bias & -0.0088 & 0.0005 & 0.0196 & -0.0008 & 3.01E-05 & 1.29E-06 & 0.0003 & -9.60E-06  \\\hline  & MSE & 9.78E-05 & 3.08E-07 & 4.10E-04 & 6.03E-07 & 2.03E-05 & 2.76E-08  & 2.10E-05 & 2.96E-08 \\\hline  & AVar & 1.89E-05 & 2.48E-08 & 1.89E-05 & 2.48E-08 & 1.89E-05 & 2.48E-08 & 1.89E-05 & 2.48E-08 \\
\hline
\multicolumn{2}{||c||}{} & \multicolumn{4}{c||}{ALSEs} & \multicolumn{4}{c||}{LSEs}\\
\hline 1 & Avg & 1.4911 & 0.5005 & 2.5184 & 0.7492 & 1.5001 & 0.4999 & 2.4992 & 0.7500 \\
\hline & Bias & -0.0089 & 0.0005  & 0.0184 & -0.0008 & 0.0001 & -1.15E-06 & -0.0007 & 2.44E-05 \\
\hline  & MSE & 1.52E-04 & 3.87E-07 & 4.21E-04 & 6.23E-07 & 8.64E-05 & 1.18E-07 & 7.82E-05 & 1.09E-07\\
\hline  & AVar & 7.56E-05 & 1.13E-07 & 7.56E-05 & 1.13E-07 & 7.56E-05 & 1.13E-07 & 7.56E-05 & 1.13E-07\\ 
\hline 
\hline
\end{tabular}}
\caption{Estimates of the parameters of model~(\ref{eq:model_1}) when errors are i.i.d. Gaussian random variables as defined in~(\ref{eq:error_normal}) and M =  N = 25}
\label{table:1}
\end{table}

\begin{table}[H]
\centering
\resizebox{0.99\textwidth}{!}{\begin{tabular}{||c|c||c|c|c|c||c|c|c|c||}
\hline
\hline
\multicolumn{2}{||c||}{Parameters} & $\alpha$ & $\beta$ & $\gamma$ & $\delta$ & $\alpha$ & $\beta$ & $\gamma$ & $\delta$ \\
\hline  \multicolumn{2}{||c||}{True values} & 1.5 & 0.5 & 2.5 & 0.75 & 1.5 & 0.5 & 2.5 & 0.75 \\
\hline $\sigma$ & & \multicolumn{4}{c||}{ALSEs} & \multicolumn{4}{c||}{LSEs} \\
\hline 0.1 & Avg & 1.5039 & 0.4999 & 2.4997 & 0.7500 & 1.5000 & 0.4999 & 2.5000 & 0.7499\\
\hline & Bias & 0.0039 & -9.22E-05 & -0.0002 & 1.01E-05 & 5.12E-06 & -6.96E-08 & 2.97E-06 & -2.76E-08 \\
\hline  & MSE & 1.95E-05 & 9.75E-09 & 3.22E-07 & 2.03E-10 & 2.73E-08 & 1.04E-11 & 3.07E-08 & 1.14E-11 \\
\hline  & AVar & 4.73E-08 & 1.77E-11 & 4.73E-08 & 1.77E-11 & 4.73E-08 & 1.77E-11 & 4.73E-08 & 1.77E-11 \\
\hline
\multicolumn{2}{||c||}{} & \multicolumn{4}{c||}{ALSEs} & \multicolumn{4}{c||}{LSEs} \\
\hline 
0.5 & Avg & 1.5041 & 0.4999 & 2.4997 & 0.7500 & 1.5000 & 0.4999 & 2.5000 & 0.7499  \\
\hline  & Bias & 0.0041 & -9.53E-05 & -0.0002 & 9.99E-06 & 2.34E-05 & -3.01E-07 & 9.23E-06 & -2.12E-07 \\
\hline  & MSE & 2.11E-05 & 1.04E-08 & 1.67E-06 & 6.70E-10 & 9.53E-07 & 3.44E-10 & 8.90E-07 & 3.33E-10 \\
\hline  & AVar & 1.18E-06 & 4.43E-10 & 1.18E-06 & 4.43E-10 & 1.18E-06 & 4.43E-10 & 1.18E-06 & 4.43E-10 \\
\hline
\multicolumn{2}{||c||}{} & \multicolumn{4}{c||}{ALSEs} & \multicolumn{4}{c||}{LSEs}\\
\hline 1 & Avg & 1.504 & 0.4999 & 2.4997 & 0.7500 & 1.5000 & 0.4999 & 2.5000 & 0.7499 \\
\hline & Bias & 0.0040 & -9.40E-05 & -0.0002 & 1.01E-05& 8.76E-05 & -1.28E-06 & 1.66E-05 & -9.05E-08 \\
\hline  & MSE & 2.42E-05 & 1.14E-08 & 5.00E-06 & 1.87E-09  & 4.24E-06 & 1.53E-09 & 4.01E-06 & 1.45E-09\\
\hline  & AVar & 4.73E-06 & 1.77E-09 & 4.73E-06 & 1.77E-09 & 4.73E-06 & 1.77E-09 & 4.73E-06 & 1.77E-09\\ 
\hline
\hline 
\end{tabular}}
\caption{Estimates of the parameters of model~(\ref{eq:model_1}) when errors are i.i.d. Gaussian random variables as defined in~(\ref{eq:error_normal}) and M =  N = 50}
\label{table:2}
\end{table}
\begin{table}[H]
\centering
\resizebox{0.99\textwidth}{!}{\begin{tabular}{||c|c||c|c|c|c||c|c|c|c||}
\hline
\hline
\multicolumn{2}{||c||}{Parameters} & $\alpha$ & $\beta$ & $\gamma$ & $\delta$ & $\alpha$ & $\beta$ & $\gamma$ & $\delta$ \\
\hline  \multicolumn{2}{||c||}{True values} & 1.5 & 0.5 & 2.5 & 0.75 & 1.5 & 0.5 & 2.5 & 0.75 \\
\hline $\sigma$ & & \multicolumn{4}{c||}{ALSEs} & \multicolumn{4}{c||}{LSEs} \\
\hline 0.1 & Avg & 1.5005 & 0.4999 & 2.4997 & 0.7500 & 1.4999 & 0.5 & 2.5000 & 0.7499\\
\hline & Bias & 0.0005 & -7.67E-06 & -0.0003 & 2.02E-06 & -4.78E-07 & 3.90E-09 & 3.55E-07 & -5.54E-09  \\
\hline  & MSE & 3.15E-07 & 6.20E-11 & 1.69E-07 & 1.68E-11 & 8.29E-09 & 1.33E-12 & 7.29E-10 & 2.04E-13 \\
\hline  & AVar & 9.34E-09 & 1.56E-12 & 9.34E-09 & 1.56E-12 & 9.34E-09 & 1.56E-12 & 9.34E-09 & 1.56E-12\\
\hline
\multicolumn{2}{||c||}{} & \multicolumn{4}{c||}{ALSEs} & \multicolumn{4}{c||}{LSEs} \\
\hline 
0.5 & Avg & 1.5003 & 0.4999 & 2.4996 & 0.7500  & 1.5000 & 0.4999 & 2.5000 & 0.7499\\
\hline  & Bias & 0.0003 & -5.85E-06 & -0.0004 & 2.40E-06 & 4.90E-06 & -1.40E-07 & 4.41E06 & -1.60E-08 \\
\hline  & MSE & 3.80E-07 & 7.20E-11 & 3.14E-07 & 4.12E-11  & 1.55E-07 & 2.62E-11 & 1.07E-07 & 1.88E-11 \\
\hline  & AVar & 2.33E-07 & 3.89E-11 & 2.33E-07 & 3.89E-11 & 2.33E-07 & 3.89E-11 & 2.33E-07 & 3.89E-11 \\
\hline
\multicolumn{2}{||c||}{} & \multicolumn{4}{c||}{ALSEs} & \multicolumn{4}{c||}{LSEs}\\
\hline 1 & Avg & 1.5004 & 0.4999 & 2.4995 & 0.7500  & 1.5000 & 0.4999 & 2.4999 & 0.7500 \\
\hline & Bias & 0.0004 & -6.70E-06  & -0.0005 & 3.89E-06 & 4.90E-05 & -6.11E-07 & -1.45E-05 & 5.86E-08\\
\hline  & MSE & 1.01E-06 & 1.73E-10 & 9.37E-07 & 1.38E-10 & 7.11E-07 & 1.17E-10 & 5.98E-07 & 9.97E-11\\
\hline  & AVar & 9.34E-07 & 1.56E-10 & 9.34E-07 & 1.56E-10 & 9.34E-07 & 1.56E-10 & 9.34E-07 & 1.56E-10 \\ 
\hline 
\hline
\end{tabular}}
\caption{Estimates of the parameters of model~(\ref{eq:model_1}) when errors are i.i.d. Gaussian random variables as defined in~(\ref{eq:error_normal}) and M =  N = 75}
\label{table:3}
\end{table}

\begin{table}[H]
\centering
\resizebox{0.99\textwidth}{!}{\begin{tabular}{||c|c||c|c|c|c||c|c|c|c||}
\hline
\hline
\multicolumn{2}{||c||}{Parameters} & $\alpha$ & $\beta$ & $\gamma$ & $\delta$ & $\alpha$ & $\beta$ & $\gamma$ & $\delta$ \\
\hline  \multicolumn{2}{||c||}{True values} & 1.5 & 0.5 & 2.5 & 0.75 & 1.5 & 0.5 & 2.5 & 0.75 \\
\hline $\sigma$ & & \multicolumn{4}{c||}{ALSEs} & \multicolumn{4}{c||}{LSEs} \\
\hline 0.1 & Avg & 1.4999 & 0.5000 & 2.4999 & 0.7500 & 1.5000 & 0.4999 & 2.5000 & 0.7500\\
\hline & Bias & -4.19E-05 & 1.99E-06 & -4.30E-05 & 3.65E-08 & 2.55E-06 & -3.13E-08 & 3.12E-07 & -5.88E-10 \\
\hline  & MSE  & 1.60E-08 & 5.19E-12 & 2.18E-08 & 1.78E-12& 5.38E-10 & 5.93E-14 & 7.86E-10 & 8.54E-14  \\
\hline  & AVar & 2.95E-09 & 2.77E-13 & 2.95E-09 & 2.77E-13 & 2.95E-09 & 2.77E-13 & 2.95E-09 & 2.77E-13\\
\hline
\multicolumn{2}{||c||}{} & \multicolumn{4}{c||}{ALSEs} & \multicolumn{4}{c||}{LSEs} \\
\hline 
0.5 & Avg & 1.4998 & 0.5000 & 2.4998 & 0.7500 & 1.5000 & 0.4999 & 2.5000 & 0.7500\\
\hline  & Bias & -0.0002 & 2.77E-06 & -0.0002 & 7.46E-07 & 4.96E-06 & -4.93E-08 & 1.34E-06 & 2.74E-08 \\
\hline  & MSE & 8.14E-08 & 1.38E-11 & 9.44E-08 & 8.14E-12  & 3.83E-08 & 3.75E-12 & 3.64E-08 & 3.66E-12\\
\hline  & AVar & 7.38E-08 & 6.92E-12 & 7.38E-08 & 6.92E-12 & 7.38E-08 & 6.92E-12 & 7.38E-08 & 6.92E-12 \\
\hline
\multicolumn{2}{||c||}{} & \multicolumn{4}{c||}{ALSEs} & \multicolumn{4}{c||}{LSEs}\\
\hline 1 & Avg  & 1.4997 & 0.5000 & 2.4997 & 0.7500 & 1.5000 & 0.4999 & 2.5000 & 0.7499 \\
\hline & Bias  & -0.0003 & 3.60E-06 & -0.0002 & 1.43E-06 & 9.37E-07 & -2.97E-08 & 2.10E-06 & -8.23E-08\\
\hline  & MSE & 2.35E-07 & 3.09E-11 & 2.71E-07 & 2.35E-11 & 1.60E-07 & 1.57E-11 & 1.91E-07 & 1.79E-11\\
\hline  & AVar & 2.95E-07 & 2.77E-11 & 2.95E-07 & 2.77E-11 & 2.95E-07 & 2.77E-11 & 2.95E-07 & 2.77E-11 \\ 
\hline 
\hline
\end{tabular}}
\caption{Estimates of the parameters of model~(\ref{eq:model_1}) when errors are i.i.d. Gaussian random variables as defined in~(\ref{eq:error_normal}) and M =  N = 100}
\label{table:4}
\end{table}

\begin{table}[H]
\centering
\resizebox{0.99\textwidth}{!}{\begin{tabular}{||c|c||c|c|c|c||c|c|c|c||}
\hline
\hline
\multicolumn{2}{||c||}{Parameters} & $\alpha$ & $\beta$ & $\gamma$ & $\delta$ & $\alpha$ & $\beta$ & $\gamma$ & $\delta$ \\
\hline  \multicolumn{2}{||c||}{True values} & 1.5 & 0.5 & 2.5 & 0.75 & 1.5 & 0.5 & 2.5 & 0.75 \\
\hline $\sigma$ & & \multicolumn{4}{c||}{ALSEs} & \multicolumn{4}{c||}{LSEs} \\
\hline 0.1 & Avg & 1.4911 & 0.5005 & 2.5193 & 0.7492& 1.4999 & 0.5000 & 2.5000 & 0.7499 \\
\hline & Bias  & -0.0089 & 0.0005 & 0.0193 & -0.0008 & -5.15E-05 & 1.81E-06 & 3.05E-05 & -7.46E-07\\
\hline  & MSE & 8.28E-05 & 2.84E-07 & 3.78E-04 & 5.60E-07 & 1.13E-06 & 1.70E-09 & 1.12E-06 & 1.57E-09\\
\hline  & AVar & 1.13E-06 & 1.70E-09 & 1.13E-06 & 1.70E-09 & 1.13E-06 & 1.70E-09 & 1.13E-06 & 1.70E-09\\
\hline
\multicolumn{2}{||c||}{} & \multicolumn{4}{c||}{ALSEs} & \multicolumn{4}{c||}{LSEs} \\
\hline 
0.5 & Avg  & 1.4910 & 0.5005 & 2.5192 & 0.7492 & 1.4998 & 0.5000 & 2.5000 & 0.7500\\
\hline  & Bias  & -0.0090 & 0.0005 & 0.0192 & -0.0007  & -0.0002 & 6.45E-06 & 2.31E-06 & 1.59E-06 \\
\hline  & MSE & 1.09E-04 & 3.29E-07 & 4.03E-04 & 5.93E-07 & 3.13E-05 & 4.60E-08 & 2.87E-05 & 4.00E-08\\
\hline  & AVar & 2.84E-05 & 4.25E-08 & 2.84E-05 & 4.25E-08 & 2.84E-05 & 4.25E-08 & 2.84E-05 & 4.25E-08 \\
\hline
\multicolumn{2}{||c||}{} & \multicolumn{4}{c||}{ALSEs} & \multicolumn{4}{c||}{LSEs}\\
\hline 1 & Avg & 1.4910 & 0.5005 & 2.5195 & 0.7492 & 1.4997 & 0.5000 & 2.5002 & 0.7499 \\
\hline & Bias & -0.0090 & 0.0005 & 0.0195 & -0.0008& -0.0003 & 8.25E-06 & 0.0002 & -6.10E-06\\
\hline  & MSE & 1.91E-04 & 4.57E-07 & 5.04E-04 & 7.30E-07 & 1.31E-04 & 1.94E-07 & 1.24E-04 & 1.77E-07 \\
\hline  & AVar & 1.13E-04 & 1.70E-07 & 1.13E-04 & 1.70E-07 & 1.13E-04 & 1.70E-07 & 1.13E-04 & 1.70E-07\\ 
\hline 
\hline
\end{tabular}}
\caption{Estimates of the parameters of model~(\ref{eq:model_1}) when errors are stationary random variables as defined in~(\ref{eq:error_MA}) and M =  N = 25}
\label{table:5}
\end{table}

\begin{table}[H]
\centering
\resizebox{0.99\textwidth}{!}{\begin{tabular}{||c|c||c|c|c|c||c|c|c|c||}
\hline
\hline
\multicolumn{2}{||c||}{Parameters} & $\alpha$ & $\beta$ & $\gamma$ & $\delta$ & $\alpha$ & $\beta$ & $\gamma$ & $\delta$ \\
\hline  \multicolumn{2}{||c||}{True values} & 1.5 & 0.5 & 2.5 & 0.75 & 1.5 & 0.5 & 2.5 & 0.75 \\
\hline $\sigma$ & & \multicolumn{4}{c||}{ALSEs} & \multicolumn{4}{c||}{LSEs} \\
\hline 0.1 & Avg & 1.5039 & 0.4999 & 2.4997 & 0.7500 & 1.5000 & 0.4999 & 2.5000 & 0.7499\\
\hline & Bias  & 0.0039 & -9.19E-05 & -0.0003 & 1.05E-05 & 1.26E-05 & -2.32E-07 & 3.69E-06 & -7.66E-08 \\
\hline  & MSE & 1.94E-05 & 9.71E-09 & 3.99E-07 & 2.32E-10 & 3.92E-08 & 1.54E-11 & 4.35E-08 & 1.60E-11 \\
\hline  & AVar  & 7.09E-08 & 2.66E-11 & 7.09E-08 & 2.66E-11 & 7.09E-08 & 2.66E-11 & 7.09E-08 & 2.66E-11\\
\hline
\multicolumn{2}{||c||}{} & \multicolumn{4}{c||}{ALSEs} & \multicolumn{4}{c||}{LSEs} \\
\hline 
0.5 & Avg & 1.5042 & 0.4999 & 2.4998 & 0.7500 & 1.5000 & 0.4999 & 2.5000 & 0.7499\\
\hline  & Bias & 0.0042 & -9.70E-05 & -0.0002 & 8.66E-06 & 6.93E-05 & -1.12E-06 & 4.43E-05 & -1.17E-06 \\
\hline  & MSE & 2.24E-05 & 1.10E-08 & 2.31E-06 & 9.16E-10 & 1.47E-06 & 5.55E-10 & 1.45E-06 & 5.63E-10 \\
\hline  & AVar  & 1.77E-06 & 6.65E-10 & 1.77E-06 & 6.65E-10 & 1.77E-06 & 6.65E-10 & 1.77E-06 & 6.65E-10 \\
\hline
\multicolumn{2}{||c||}{} & \multicolumn{4}{c||}{ALSEs} & \multicolumn{4}{c||}{LSEs}\\
\hline 1 & Avg & 1.5041 & 0.4999 & 2.4998 & 0.7500 & 1.4999 & 0.5000 & 2.4999 & 0.7499\\
\hline & Bias & 0.0041 & -9.59E-05 & -0.0002 & 8.50E-06 & -3.56E-05  & 1.71E-07 & -2.04E-05 & -1.20E-07 \\
\hline  & MSE & 2.60E-05 & 1.24E-08 & 7.63E-06 & 2.77E-09 & 6.11E-06 & 2.30E-09 & 6.68E-06 & 2.37E-09\\
\hline  & AVar  & 7.09E-06 & 2.66E-09 & 7.09E-06 & 2.66E-09 & 7.09E-06 & 2.66E-09 &  7.09E-06  & 2.66E-09 \\ 
\hline 
\hline
\end{tabular}}
\caption{Estimates of the parameters of model~(\ref{eq:model_1}) when errors are stationary random variables as defined in~(\ref{eq:error_MA}) and M =  N = 50}
\label{table:6}
\end{table}
\begin{table}[H]
\centering
\resizebox{0.99\textwidth}{!}{\begin{tabular}{||c|c||c|c|c|c||c|c|c|c||}
\hline
\hline
\multicolumn{2}{||c||}{Parameters} & $\alpha$ & $\beta$ & $\gamma$ & $\delta$ & $\alpha$ & $\beta$ & $\gamma$ & $\delta$ \\
\hline  \multicolumn{2}{||c||}{True values} & 1.5 & 0.5 & 2.5 & 0.75 & 1.5 & 0.5 & 2.5 & 0.75 \\
\hline $\sigma$ & & \multicolumn{4}{c||}{ALSEs} & \multicolumn{4}{c||}{LSEs} \\
\hline 0.1 & Avg  & 1.5005 & 0.4999 & 2.4997  & 0.7500 & 1.4999 & 0.5000 & 2.5000 & 0.7499\\
\hline & Bias  & 0.0005 & -7.49E-06 &  -0.0002 &  1.92E-06 & -1.18E-06 & 1.58E-08 & 1.46E-06 & -1.12E-08 \\
\hline  & MSE & 3.11E-07 & 6.12E-11 & 1.68E-07 &  1.71E-11  & 1.21E-08 & 2.03E-12 & 9.24E-10 & 2.82E-13 \\
\hline  & AVar & 1.40E-08 & 2.33E-12 & 1.40E-08  & 2.33E-12 & 1.40E-08 &  2.33E-12 & 1.40E-08 & 2.33E-12\\
\hline
\multicolumn{2}{||c||}{} & \multicolumn{4}{c||}{ALSEs} & \multicolumn{4}{c||}{LSEs} \\
\hline 
0.5 & Avg & 1.5004 &  0.4999 &  2.4996 & 0.7500 & 1.5000 & 0.4999 & 2.5000 & 0.7500 \\
\hline  & Bias & 0.0004 & -6.10E-06 &  -0.0004 & 3.16E-06  & 5.48E-06 & -8.28E-08 & 1.95E-06 &  4.45E-08 \\
\hline  & MSE & 5.07E-07 & 9.31E-11 & 4.75E-07 & 6.37E-11 & 2.80E-07 & 2.80E-07 & 2.10E-07 & 3.49E-11\\
\hline  & AVar  & 3.50E-07 & 5.83E-11 & 3.50E-07 & 5.83E-11 & 3.50E-07 & 5.83E-11 & 3.50E-07 & 5.83E-11 \\
\hline
\multicolumn{2}{||c||}{} & \multicolumn{4}{c||}{ALSEs} & \multicolumn{4}{c||}{LSEs}\\
\hline 1 & Avg & 1.5004 & 0.4999 & 2.4995 & 0.7500 & 1.5000 & 0.4999 &  2.4999 & 0.7500 \\
\hline & Bias & 0.0004 & -6.65E-06 &  -0.0005 & 4.30E-06 & 3.76E-05  & -6.39E-07 &  -1.91E-05 & 1.19E-07\\
\hline  & MSE & 1.37E-06 & 2.39E-10 & 1.26E-06 &1.90E-10& 1.07E-06 & 1.80E-10 & 9.31E-07 & 1.58E-10 \\
\hline  & AVar & 1.40E-06 & 2.33E-10 & 1.40E-06 & 2.33E-10 & 1.40E-06 & 2.33E-10 & 1.40E-06 & 2.33E-10\\ 
\hline
\hline 
\end{tabular}}
\caption{Estimates of the parameters of model~(\ref{eq:model_1})when errors are stationary random variables as defined in~(\ref{eq:error_MA}) and M =  N = 75}
\label{table:7}
\end{table}

\begin{table}[H]
\centering
\resizebox{0.99\textwidth}{!}{\begin{tabular}{||c|c||c|c|c|c||c|c|c|c||}
\hline
\hline
\multicolumn{2}{||c||}{Parameters} & $\alpha$ & $\beta$ & $\gamma$ & $\delta$ & $\alpha$ & $\beta$ & $\gamma$ & $\delta$ \\
\hline  \multicolumn{2}{||c||}{True values} & 1.5 & 0.5 & 2.5 & 0.75 & 1.5 & 0.5 & 2.5 & 0.75 \\
\hline $\sigma$ & & \multicolumn{4}{c||}{ALSEs} & \multicolumn{4}{c||}{LSEs} \\
\hline 0.1 & Avg & 1.4999 & 0.5000 & 2.4999 & 0.7500& 1.5000 & 0.4999 & 2.4999 & 0.7500\\
\hline & Bias  & -4.14E-05 & 1.98E-06 & -4.85E-05 & 9.25E-08 & 3.60E-06 & -3.84E-08 & -9.42E-07 & 1.36E-08 \\
\hline  & MSE & 1.68E-08 & 5.28E-12 & 2.5063E-08 & 2.02E-12 & 9.26E-10 & 1.07E-13 & 1.81E-09 &   1.82E-13\\
\hline  & AVar & 4.43E-09 & 4.15E-13 & 4.43E-09 & 4.15E-13 & 4.43E-09 & 4.15E-13 & 4.43E-09 & 4.15E-13\\
\hline
\multicolumn{2}{||c||}{} & \multicolumn{4}{c||}{ALSEs} & \multicolumn{4}{c||}{LSEs} \\
\hline 
0.5 & Avg  & 1.4998 & 0.5000 & 2.4998 & 0.7500 & 1.4999 & 0.5000 & 2.4999 & 0.7500\\
\hline  & Bias &  -0.0002 & 3.21E-06 & -0.0001 & 1.02E-06 & -6.40E-06 & 3.26E-08 & -4.78E-06 &  3.26E-08 \\
\hline  & MSE & 1.36E-07 & 2.00E-11 & 1.36E-07 & 1.16E-11 & 6.31E-08 & 6.15E-12 & 6.12E-08 & 5.81E-12\\
\hline  & AVar  & 1.11E-07 & 1.04E-11 & 1.11E-07 & 1.04E-11 & 1.11E-07 & 1.04E-11 & 1.11E-07 & 1.04E-11 \\
\hline
\multicolumn{2}{||c||}{} & \multicolumn{4}{c||}{ALSEs} & \multicolumn{4}{c||}{LSEs}\\
\hline 1 & Avg & 1.4997 & 0.5000 & 2.4997 & 0.7500 & 1.4999 & 0.5000 & 2.5000 & 0.7499 \\
\hline & Bias & -0.0003 & 3.94E-06 & -0.0003 & 1.60E-06 & -2.75E-05 & 2.64E-07 & 6.40E-06 & -4.03E-08 \\
\hline  & MSE & 3.66E-07 & 4.48E-11 & 3.67E-07 &   3.29E-11 & 2.73E-07 & 2.67E-11 & 2.78E-07 & 2.67E-11\\
\hline  & AVar  & 4.43E-07 & 4.15E-11 & 4.43E-07 &   4.15E-11  & 4.43E-07 & 4.15E-11 & 4.43E-07 &   4.15E-11 \\ 
\hline 
\hline
\end{tabular}}
\caption{Estimates of the parameters of model~(\ref{eq:model_1}) when errors are stationary random variables as defined in~(\ref{eq:error_MA}) and M =  N = 100}
\label{table:8}
\end{table}
\subsection{Simulation results for the multiple component model with $\mathbf{p = 2}$}
Next we conduct numerical simulations on model~(\ref{eq:model_mul_comp}) with $p = 2$ and the following parameters:
$$A_1^0 = 5,\ B_1^0 = 4,\ \alpha_1^0 = 2.1,\ \beta_1^0 = 0.1,\ \gamma_1^0 = 1.25\ \textmd{and}\ \delta_1^0 = 0.25. $$
$$ A_2^0 = 3,\ B_2^0 = 2,\ \alpha_2^0 = 1.5 ,\ \beta_2^0 = 0.5,\ \gamma_2^0 = 1.75\ \textmd{and}\ \delta_2^0 = 0.75$$
The error structures used to generate the data are same as that used for the one component model, see equations,~(\ref{eq:error_normal}) and (\ref{eq:error_MA}).
For simulations we consider different values of $\sigma$ and different values of $M$ and $N$, again same as that for the one component model.
We estimate the parameters both by least  squares estimation method and approximate least squares estimation method. These estimates are obtained 1000 times each and averages, biases, MSEs and asymptotic variances are computed. The results are reported in the following tables.
From the tables, it can be seen that the estimates, both the ALSEs and the LSEs are quite close to their true values. It is observed that the estimates of the second component are better than those of the first component, in the sense that their biases and MSEs are smaller and the MSEs are better matched with the corresponding asymptotic variances. For both the estimators, as the sample size increases, the MSEs and the biases of the estimates of both components, decrease thus showing consistency.

\begin{table}[!htbp]
\centering
\resizebox{\textwidth}{!}{\begin{tabular}{||c|c||c|c|c|c|c|c|c|c||}
\hline
\multicolumn{2}{||c||}{Parameters}  & $\alpha_1$ & $\beta_1$ & $\gamma_1$ & $\delta_1$ & $\alpha_2$ & $\beta_2$ & $\gamma_2$ & $\delta_2$ \\ \hline
\multicolumn{2}{||c||}{True values} & 2.1        & 0.1       & 1.25       & 0.25       & 1.5        & 0.5       & 1.75       & 0.75       \\ \hline
$\sigma$         & & \multicolumn{8}{c||}{}                                                                                                \\ \hline
    &         & \multicolumn{8}{c||}{ALSEs}                                                               \\ \hline
    & Average & 2.1154   & 0.0994   & 1.2587   & 0.2500   & 1.5411   & 0.4988    & 1.7664   & 0.7493    \\ \hline
    & Bias    & 0.0154   & -0.0006  & 0.0087   & 1.01E-05 & 0.0411   & -0.0012   & 0.0164   & -0.0007   \\ \hline
    & MSE     & 2.36E-04 & 3.48E-07 & 7.67E-05 & 4.85E-10 & 2.36E-04 & 1.45E-06  & 2.68E-04 & 4.85E-10  \\ \hline
0.1 &         & \multicolumn{8}{c||}{LSEs}                                                                \\ \hline
    & Average & 2.1031   & 0.0998   & 1.2565   & 0.2500   & 1.5017   & 0.5000    & 1.7510   & 0.7500    \\ \hline
    & Bias    & 0.0031   & -0.0002  & 0.0065   & 3.83E-05 & 0.0017   & -2.16E-05 & 0.0010   & -2.92E-05 \\ \hline
    & MSE     & 9.70E-06 & 3.14E-08 & 4.23E-05 & 1.85E-09 & 3.71E-06 & 1.75E-09  & 1.93E-06 & 2.11E-09  \\ \hline
    &         & \multicolumn{8}{l||}{}                                                                    \\ \hline
    & AVar    & 2.40E-07 & 3.60E-10 & 2.40E-07 & 3.60E-10 & 7.56E-07 & 1.13E-09  & 7.56E-07 & 1.13E-09  \\ \hline
    &         & \multicolumn{8}{c||}{ALSEs}                                                               \\ \hline
    & Average & 2.1154   & 0.0994   & 1.2586   & 0.2500   & 1.5412   & 0.4988    & 1.7664   & 0.7493    \\ \hline
    & Bias    & 0.0154   & -0.0006  & 0.0086   & 1.49E-05 & 0.0412   & -0.0012   & 0.0164   & -0.0007   \\ \hline
    & MSE     & 2.44E-04 & 3.59E-07 & 8.02E-05 & 8.99E-09 & 2.44E-04 & 1.48E-06  & 2.87E-04 & 8.99E-09  \\ \hline
0.5 &         & \multicolumn{8}{c||}{LSEs}                                                                \\ \hline
    & Average & 2.1031   & 0.0998   & 1.2563   & 0.2500   & 1.5017   & 0.5000    & 1.7510   & 0.7500    \\ \hline
    & Bias    & 0.0031   & -0.0002  & 0.0063   & 4.40E-05 & 0.0017   & -2.25E-05 & 0.0010   & -3.13E-05 \\ \hline
    & MSE     & 1.66E-05 & 4.03E-08 & 4.63E-05 & 1.04E-08 & 2.48E-05 & 3.16E-08  & 2.55E-05 & 3.50E-08  \\ \hline
    &         & \multicolumn{8}{l||}{}                                                                    \\ \hline
    & AVar    & 5.99E-06 & 8.99E-09 & 5.99E-06 & 8.99E-09 & 1.89E-05 & 2.84E-08  & 1.89E-05 & 2.84E-08  \\ \hline
    &         & \multicolumn{8}{c||}{ALSEs}                                                               \\ \hline
    & Average & 2.1154   & 0.0994   & 1.2585   & 0.2500   & 1.5408   & 0.4988    & 1.7665   & 0.7493    \\ \hline
    & Bias    & 0.0154   & -0.0006  & 0.0085   & 1.88E-05 & 0.0408   & -0.0012   & 0.0165   & -0.0007   \\ \hline
    & MSE     & 2.65E-04 & 3.84E-07 & 9.75E-05 & 3.93E-08 & 2.65E-04 & 1.53E-06  & 3.38E-04 & 3.93E-08  \\ \hline
1   &         & \multicolumn{8}{c||}{LSEs}                                                                \\ \hline
    & Average & 2.1031   & 0.0998   & 1.2563   & 0.2500   & 1.5015   & 0.5000    & 1.7513   & 0.7500    \\ \hline
    & Bias    & 0.0031   & -0.0002  & 0.0063   & 4.78E-05 & 0.0015   & -1.40E-05 & 0.0013   & -4.21E-05 \\ \hline
    & MSE     & 3.63E-05 & 6.50E-08 & 6.50E-05 & 3.98E-08 & 8.57E-05 & 1.22E-07  & 8.44E-05 & 1.18E-07  \\ \hline
    &         & \multicolumn{8}{l||}{}                                                                    \\ \hline
    & AVar    & 2.40E-05 & 3.60E-08 & 2.40E-05 & 3.60E-08 & 7.56E-05 & 1.13E-07  & 7.56E-05 & 1.13E-07  \\ \hline
\end{tabular}}
\caption{Estimates of the parameters of model~(\ref{eq:model_mul_comp}) when errors are i.i.d Gaussian random variables as defined in~(\ref{eq:error_normal}) and M =  N = 25}
\label{table:9}
\end{table}
\begin{table}[p]
\centering
\resizebox{\textwidth}{!}{\begin{tabular}{||c|c||c|c|c|c|c|c|c|c||}
\hline
\multicolumn{2}{||c||}{Parameters}  & $\alpha_1$ & $\beta_1$ & $\gamma_1$ & $\delta_1$ & $\alpha_2$ & $\beta_2$ & $\gamma_2$ & $\delta_2$ \\ \hline
\multicolumn{2}{||c||}{True values} & 2.1        & 0.1       & 1.25       & 0.25       & 1.5        & 0.5       & 1.75       & 0.75       \\ \hline
$\sigma$         & & \multicolumn{8}{c||}{}                                                                                                \\ \hline
    &         & \multicolumn{8}{c||}{ALSEs}                                                                \\ \hline
    & Average & 2.1011   & 0.1000    & 1.2597   & 0.2499   & 1.5127   & 0.4997    & 1.7529   & 0.7499    \\ \hline
    & Bias    & 0.0011   & -1.36E-05 & 0.0097   & -0.0001  & 0.0127   & -0.0003   & 0.0029   & -5.92E-05 \\ \hline
    & MSE     & 1.16E-06 & 1.92E-10  & 9.37E-05 & 2.07E-08 & 1.16E-06 & 7.47E-08  & 8.34E-06 & 2.07E-08  \\ \hline
0.1 &         & \multicolumn{8}{c||}{LSEs}                                                                 \\ \hline
    & Average & 2.1010   & 0.1000    & 1.2572   & 0.2499   & 1.5007   & 0.5000    & 1.7507   & 0.7500    \\ \hline
    & Bias    & 0.0010   & -1.07E-05 & 0.0072   & -0.0001  & 0.0007   & -1.35E-05 & 0.0007   & -1.19E-05 \\ \hline
    & MSE     & 1.12E-06 & 1.24E-10  & 5.18E-05 & 1.49E-08 & 6.03E-07 & 1.99E-10  & 5.16E-07 & 1.60E-10  \\ \hline
    &         & \multicolumn{8}{l||}{}                                                                     \\ \hline
    & AVar    & 1.50E-08 & 5.62E-12  & 1.50E-08 & 5.62E-12 & 4.73E-08 & 1.77E-11  & 4.73E-08 & 1.77E-11  \\ \hline
    &         & \multicolumn{8}{c||}{ALSEs}                                                                \\ \hline
    & Average & 2.1011   & 0.1000    & 1.2597   & 0.2499   & 1.5127   & 0.4997    & 1.7529   & 0.7499    \\ \hline
    & Bias    & 0.0011   & -1.36E-05 & 0.0097   & -0.0001  & 0.0127   & -0.0003   & 0.0029   & -5.94E-05 \\ \hline
    & MSE     & 1.57E-06 & 3.33E-10  & 9.39E-05 & 2.08E-08 & 1.57E-06 & 7.46E-08  & 9.66E-06 & 2.08E-08  \\ \hline
0.5 &         & \multicolumn{8}{c||}{LSEs}                                                                 \\ \hline
    & Average & 2.1011   & 0.1000    & 1.2572   & 0.2499   & 1.5007   & 0.5000    & 1.7507   & 0.7500    \\ \hline
    & Bias    & 0.0011   & -1.09E-05 & 0.0072   & -0.0001  & 0.0007   & -1.27E-05 & 0.0007   & -1.20E-05 \\ \hline
    & MSE     & 1.53E-06 & 2.59E-10  & 5.22E-05 & 1.50E-08 & 1.75E-06 & 5.97E-10  & 1.67E-06 & 5.74E-10  \\ \hline
    &         & \multicolumn{8}{l||}{}                                                                     \\ \hline
    & AVar    & 3.75E-07 & 1.40E-10  & 3.75E-07 & 1.40E-10 & 1.18E-06 & 4.43E-10  & 1.18E-06 & 4.43E-10  \\ \hline
    &         & \multicolumn{8}{c||}{ALSEs}                                                                \\ \hline
    & Average & 2.1010   & 0.1000    & 1.2597   & 0.2499   & 1.5127   & 0.4997    & 1.7528   & 0.7499    \\ \hline
    & Bias    & 0.0010   & -1.32E-05 & 0.0097   & -0.0001  & 0.0127   & -0.0003   & 0.0028   & -5.66E-05 \\ \hline
    & MSE     & 2.69E-06 & 7.54E-10  & 9.51E-05 & 2.13E-08 & 2.69E-06 & 7.60E-08  & 1.28E-05 & 2.13E-08  \\ \hline
1   &         & \multicolumn{8}{c||}{LSEs}                                                                 \\ \hline
    & Average & 2.1010   & 0.1000    & 1.2572   & 0.2499   & 1.5007   & 0.5000    & 1.7506   & 0.7500    \\ \hline
    & Bias    & 0.0010   & -1.03E-05 & 0.0072   & -0.0001  & 0.0007   & -1.30E-05 & 0.0006   & -9.31E-06 \\ \hline
    & MSE     & 2.62E-06 & 6.72E-10  & 5.32E-05 & 1.54E-08 & 5.14E-06 & 1.84E-09  & 5.11E-06 & 1.80E-09  \\ \hline
    &         & \multicolumn{8}{l||}{}                                                                     \\ \hline
    & AVar    & 1.50E-06 & 5.62E-10  & 1.50E-06 & 5.62E-10 & 4.73E-06 & 1.77E-09  & 4.73E-06 & 1.77E-09  \\ \hline
\end{tabular}}
\caption{Estimates of the parameters of model~(\ref{eq:model_mul_comp}) when errors are i.i.d Gaussian random variables as defined in~(\ref{eq:error_normal}) and M =  N = 50}
\label{table:10}
\end{table}
\begin{table}[p]
\centering
\resizebox{\textwidth}{!}{\begin{tabular}{||c|c||c|c|c|c|c|c|c|c||}
\hline
\multicolumn{2}{||c||}{Parameters}  & $\alpha_1$ & $\beta_1$ & $\gamma_1$ & $\delta_1$ & $\alpha_2$ & $\beta_2$ & $\gamma_2$ & $\delta_2$ \\ \hline
\multicolumn{2}{||c||}{True values} & 2.1        & 0.1       & 1.25       & 0.25       & 1.5        & 0.5       & 1.75       & 0.75       \\ \hline
$\sigma$         & & \multicolumn{8}{c||}{}    \\ \hline
    &         & \multicolumn{8}{c||}{ALSEs} \\ \hline
    & Average & 2.0999    & 0.1000      & 1.2534      & 0.2500      & 1.5002      & 0.5000      & 1.7506      & 0.7500      \\ \hline
    & Bias    & -6.54E-05 & -4.88E-07   & 0.0034      & -4.26E-05   & 0.0002      & -2.87E-06   & 0.0006      & -6.20E-06   \\ \hline
    & MSE     & 7.71E-09  & 7.86E-13    & 1.16E-05    & 1.82E-09    & 7.71E-09    & 9.76E-12    & 3.97E-07    & 1.82E-09    \\ \hline
0.1 &         & \multicolumn{8}{c||}{LSEs}    \\ \hline
    & Average & 2.1000    & 0.1000      & 1.2528      & 0.2500      & 1.5001      & 0.5000      & 1.7500      & 0.7500      \\ \hline
    & Bias    & -2.64E-05 & 6.85E-08    & 0.0028      & -3.27E-05   & 5.99E-05    & -6.06E-07   & 2.94E-05    & -1.04E-07   \\ \hline
    & MSE     & 4.16E-09  & 5.63E-13    & 7.86E-06    & 1.07E-09    & 1.27E-08    & 1.83E-12    & 9.57E-09    & 1.40E-12    \\ \hline
    &         & \multicolumn{8}{l||}{}   \\ \hline
    & AVar    & 2.96E-09  & 4.93E-13    & 2.96E-09    & 4.93E-13    & 9.34E-09    & 1.56E-12    & 9.34E-09    & 1.56E-12    \\ \hline
    &         & \multicolumn{8}{c||}{ALSEs} \\ \hline
    & Average & 2.0999    & 0.1000      & 1.2534      & 0.2500      & 1.5001      & 0.5000      & 1.7506      & 0.7500      \\ \hline
    & Bias    & -7.54E-05 & -3.69E-07   & 0.0034      & -4.26E-05   & 0.0001      & -2.16E-06   & 0.0006      & -6.22E-06   \\ \hline
    & MSE     & 8.85E-08  & 1.35E-11    & 1.16E-05    & 1.83E-09    & 8.85E-08    & 3.64E-11    & 6.31E-07    & 1.83E-09    \\ \hline
0.5 &         & \multicolumn{8}{c||}{LSEs}  \\ \hline
    & Average & 2.1000    & 0.1000      & 1.2528      & 0.2500      & 1.5000      & 0.5000      & 1.7500      & 0.7500      \\ \hline
    & Bias    & -3.61E-05 & 1.81E-07    & 0.0028      & -3.27E-05   & 3.47E-05    & -2.27E-07   & 2.33E-05    & -4.46E-08   \\ \hline
    & MSE     & 8.36E-08  & 1.34E-11    & 7.91E-06    & 1.08E-09    & 1.91E-07    & 3.11E-11    & 2.54E-07    & 4.18E-11    \\ \hline
    &         & \multicolumn{8}{l||}{} \\ \hline
    & AVar    & 7.40E-08  & 1.23E-11    & 7.40E-08    & 1.23E-11    & 2.33E-07    & 3.89E-11    & 2.33E-07    & 3.89E-11    \\ \hline
    &         & \multicolumn{8}{c||}{ALSEs} \\ \hline
    & Average & 2.0999 & 0.0999 & 1.2534 & 0.2499 & 1.5001 & 0.4999 & 1.7506 & 0.7499 \\ \hline
    & Bias    & -4.69E-05 & -8.01E-07   & 0.00341 & -4.27E-05   & 0.0001 & -2.93E-06   & 0.0006 & -6.44E-06   \\ \hline
    & MSE     & 3.07E-07  & 4.87E-11    & 1.20E-05    & 1.88E-09    & 3.07E-07    & 1.58E-10    & 1.41E-06    & 1.88E-09    \\ \hline
1   &         & \multicolumn{8}{c||}{LSEs}
\\ \hline
    & Average & 2.1000    & 0.1000      & 1.2528      & 0.2500      & 1.5001      & 0.5000      & 1.7500      & 0.7500      \\ \hline
    & Bias    & -9.67E-06 & -2.20E-07   & 0.0028      & -3.28E-05   & 7.50E-05    & -1.01E-06   & 4.89E-05    & -2.68E-07   \\ \hline
    & MSE     & 2.94E-07  & 4.70E-11    & 8.24E-06    & 1.13E-09    & 8.96E-07    & 1.45E-10    & 1.02E-06    & 1.60E-10    \\ \hline
    &         & \multicolumn{8}{l||}{}
    \\ \hline
    & AVar    & 2.96E-07  & 4.93E-11    & 2.96E-07    & 4.93E-11    & 9.34E-07    & 1.56E-10    & 9.34E-07    & 1.56E-10    \\ \hline
\end{tabular}}
\caption{Estimates of the parameters of model~(\ref{eq:model_mul_comp}) when errors are i.i.d Gaussian random variables as defined in~(\ref{eq:error_normal}) and M =  N = 75}
\label{table:11}
\end{table}
\begin{table}[p]
\centering
\resizebox{\textwidth}{!}{\begin{tabular}{||c|c||c|c|c|c|c|c|c|c||}
\hline
\multicolumn{2}{||c||}{Parameters}  & $\alpha_1$ & $\beta_1$ & $\gamma_1$ & $\delta_1$ & $\alpha_2$ & $\beta_2$ & $\gamma_2$ & $\delta_2$ \\ \hline
\multicolumn{2}{||c||}{True values} & 2.1        & 0.1       & 1.25       & 0.25       & 1.5        & 0.5       & 1.75       & 0.75       \\ \hline
$\sigma$         & & \multicolumn{8}{c||}{}                                                                                                \\ \hline
    &         & \multicolumn{8}{c||}{ALSEs}                                                                \\ \hline
    & Average & 2.1005   & 0.1000    & 1.2502   & 0.2500    & 1.4991   & 0.5000    & 1.7501    & 0.7500    \\ \hline
    & Bias    & 0.0005   & -6.19E-06 & 0.0002   & -3.44E-06 & -0.0009  & 1.31E-05  & 7.84E-05  & -1.07E-07 \\ \hline
    & MSE     & 3.00E-07 & 3.85E-11  & 6.07E-08 & 1.20E-11  & 3.00E-07 & 1.72E-10  & 9.58E-09  & 1.20E-11  \\ \hline
0.1 &         & \multicolumn{8}{c||}{LSEs}                                                                   \\ \hline
    & Average & 2.0995   & 0.1000    & 1.2504   & 0.2500    & 1.5000   & 0.5000    & 1.7500    & 0.7500    \\ \hline
    & Bias    & -0.0005  & 2.67E-06  & 0.0004   & -3.55E-06 & 8.50E-06 & -6.72E-08 & -2.24E-06 & -2.87E-08 \\ \hline
    & MSE     & 2.19E-07 & 7.29E-12  & 1.99E-07 & 1.27E-11  & 2.81E-09 & 2.56E-13  & 2.64E-09  & 2.49E-13  \\ \hline
    &         & \multicolumn{8}{c||}{}                                                                       \\ \hline
    & AVar    & 9.37E-10 & 8.78E-14  & 9.37E-10 & 8.78E-14  & 2.95E-09 & 2.77E-13  & 2.95E-09  & 2.77E-13  \\ \hline
    &         & \multicolumn{8}{c||}{ALSEs}                                                                  \\ \hline
    & Average & 2.1006   & 0.1000    & 1.2502   & 0.2500    & 1.4991   & 0.5000    & 1.7501    & 0.7500    \\ \hline
    & Bias    & 0.0006   & -6.24E-06 & 0.0002   & -3.37E-06 & -0.0009  & 1.31E-05  & 7.87E-05  & -1.43E-07 \\ \hline
    & MSE     & 3.31E-07 & 4.13E-11  & 7.89E-08 & 1.35E-11  & 3.31E-07 & 1.78E-10  & 7.61E-08  & 1.35E-11  \\ \hline
0.5 &         & \multicolumn{8}{c||}{LSEs}                                                                   \\ \hline
    & Average & 2.0995   & 0.1000    & 1.2504   & 0.2500    & 1.5000   & 0.5000    & 1.7500    & 0.7500    \\ \hline
    & Bias    & -0.0005  & 2.65E-06  & 0.0004   & -3.50E-06 & 7.81E-06 & -5.19E-08 & -2.02E-06 & -6.13E-08 \\ \hline
    & MSE     & 2.39E-07 & 9.24E-12  & 2.17E-07 & 1.44E-11  & 7.26E-08 & 6.57E-12  & 7.23E-08  & 6.81E-12  \\ \hline
    &         & \multicolumn{8}{l||}{}                                                                       \\ \hline
    & AVar    & 2.34E-08 & 2.20E-12  & 2.34E-08 & 2.20E-12  & 7.38E-08 & 6.92E-12  & 7.38E-08  & 6.92E-12  \\ \hline
    &         & \multicolumn{8}{c||}{ALSEs}                                                                  \\ \hline
    & Average & 2.1005   & 0.1000    & 1.2502   & 0.2500    & 1.4991   & 0.5000    & 1.7501    & 0.7500    \\ \hline
    & Bias    & 0.0005   & -6.21E-06 & 0.0002   & -3.55E-06 & -0.0009  & 1.30E-05  & 0.0001    & -3.75E-07 \\ \hline
    & MSE     & 4.01E-07 & 4.76E-11  & 1.51E-07 & 2.08E-11  & 4.01E-07 & 1.98E-10  & 3.15E-07  & 2.08E-11  \\ \hline
1   &         & \multicolumn{8}{c||}{LSEs}                                                                   \\ \hline
    & Average & 2.0995   & 0.1000    & 1.2504   & 0.2500    & 1.5000   & 0.5000    & 1.7500    & 0.7500    \\ \hline
    & Bias    & -0.0005  & 2.67E-06  & 0.0004   & -3.65E-06 & 1.66E-05 & -1.33E-07 & 3.24E-05  & -2.92E-07 \\ \hline
    & MSE     & 3.14E-07 & 1.59E-11  & 2.93E-07 & 2.17E-11  & 3.28E-07 & 2.98E-11  & 3.16E-07  & 2.85E-11  \\ \hline
    &         & \multicolumn{8}{l||}{}                                                                       \\ \hline
    & AVar    & 9.37E-08 & 8.78E-12  & 9.37E-08 & 8.78E-12  & 2.95E-07 & 2.77E-11  & 2.95E-07  & 2.77E-11  \\ \hline
\end{tabular}}
\caption{Estimates of the parameters of model~(\ref{eq:model_mul_comp}) when errors are i.i.d Gaussian random variables as defined in~(\ref{eq:error_normal}) and M =  N = 100}
\label{table:12}
\end{table}
\begin{table}[p]
\centering
\resizebox{\textwidth}{!}{\begin{tabular}{||c|c||c|c|c|c|c|c|c|c||}
\hline
\multicolumn{2}{||c||}{Parameters}  & $\alpha_1$ & $\beta_1$ & $\gamma_1$ & $\delta_1$ & $\alpha_2$ & $\beta_2$ & $\gamma_2$ & $\delta_2$ \\ \hline
\multicolumn{2}{||c||}{True values} & 2.1        & 0.1       & 1.25       & 0.25       & 1.5        & 0.5       & 1.75       & 0.75       \\ \hline
$\sigma$         & & \multicolumn{8}{c||}{}                                                                                                \\ \hline
               &  & \multicolumn{8}{c||}{ALSEs}                                                                                           \\ \hline
                 & Average        & 2.1153     & 0.0994    & 1.2588     & 0.2500     & 1.5411     & 0.4988    & 1.7663     & 0.7493     \\ \hline
                 & Bias           & 0.0153     & -0.0006   & 0.0088     & 9.19E-06   & 0.0411     & -0.0012   & 0.0163     & -0.0007    \\ \hline
                 & MSE            & 2.36E-04   & 3.48E-07  & 7.72E-05   & 6.90E-10   & 2.36E-04   & 1.46E-06  & 2.67E-04   & 6.90E-10   \\ \hline
0.1            &  & \multicolumn{8}{c||}{LSEs}                                                                                            \\ \hline
                 & Average        & 2.1030     & 0.0998    & 1.2565     & 0.2500     & 1.5017     & 0.5000    & 1.7510     & 0.7500     \\ \hline
                 & Bias           & 0.0030     & -0.0002   & 0.0065     & 3.74E-05   & 0.0017     & -2.29E-05 & 0.0010     & -2.70E-05  \\ \hline
                 & MSE            & 9.68E-06   & 3.14E-08  & 4.26E-05   & 1.95E-09   & 4.09E-06   & 2.31E-09  & 2.34E-06   & 2.79E-09   \\ \hline
              &   & \multicolumn{8}{c||}{}                                                                                                \\ \hline
                 & Avar           & 3.60E-07   & 5.39E-10  & 3.60E-07   & 5.39E-10   & 1.13E-06   & 1.70E-09  & 1.13E-06   & 1.70E-09   \\ \hline
               &  & \multicolumn{8}{c||}{ALSEs}                                                                                           \\ \hline
                 & Average        & 2.1153     & 0.0994    & 1.2587     & 0.2500     & 1.5411     & 0.4988    & 1.7663     & 0.7493     \\ \hline
                 & Bias           & 0.0153     & -0.0006   & 0.0087     & 1.01E-05   & 0.0411     & -0.0012   & 0.0163     & -0.0007    \\ \hline
                 & MSE            & 2.44E-04   & 3.59E-07  & 8.41E-05   & 1.44E-08   & 2.44E-04   & 1.49E-06  & 2.95E-04   & 1.44E-08   \\ \hline
0.5           &   & \multicolumn{8}{c||}{LSEs}                                                                                            \\ \hline
                 & Average        & 2.1030     & 0.0998    & 1.2564     & 0.2500     & 1.5017     & 0.5000    & 1.7510     & 0.7500     \\ \hline
                 & Bias           & 0.0030     & -0.0002   & 0.0064     & 3.94E-05   & 0.0017     & -2.32E-05 & 0.0010     & -2.96E-05  \\ \hline
                 & MSE            & 1.77E-05   & 4.27E-08  & 4.97E-05   & 1.49E-08   & 3.27E-05   & 4.57E-08  & 3.68E-05   & 5.16E-08   \\ \hline
                &  & \multicolumn{8}{c||}{}                                                                                                \\ \hline
                 & Avar           & 8.99E-06   & 1.35E-08  & 8.99E-06   & 1.35E-08   & 2.84E-05   & 4.25E-08  & 2.84E-05   & 4.25E-08   \\ \hline
             &    & \multicolumn{8}{c||}{ALSEs}                                                                                           \\ \hline
                 & Average        & 2.1158     & 0.0994    & 1.2586     & 0.2500     & 1.5412     & 0.4988    & 1.7666     & 0.7493     \\ \hline
                 & Bias           & 0.0158     & -0.0006   & 0.0086     & 1.76E-05   & 0.0412     & -0.0012   & 0.0166     & -0.0007    \\ \hline
                 & MSE            & 2.88E-04   & 4.21E-07  & 1.06E-04   & 5.96E-08   & 2.88E-04   & 1.62E-06  & 3.88E-04   & 5.96E-08   \\ \hline
1              &  & \multicolumn{8}{c||}{LSEs}                                                                                            \\ \hline
                 & Average        & 2.1035     & 0.0998    & 1.2562     & 0.2500     & 1.5019     & 0.5000    & 1.7515     & 0.7500     \\ \hline
                 & Bias           & 0.0035     & -0.0002   & 0.0062     & 4.79E-05   & 0.0019     & -2.81E-05 & 0.0015     & -4.42E-05  \\ \hline
                 & MSE            & 4.65E-05   & 8.66E-08  & 7.20E-05   & 5.64E-08   & 1.27E-04   & 1.88E-07  & 1.40E-04   & 2.07E-07   \\ \hline
               &  & \multicolumn{8}{c||}{}                                                                                                \\ \hline
                 & Avar           & 3.60E-05   & 5.39E-08  & 3.60E-05   & 5.39E-08   & 1.13E-04   & 1.70E-07  & 1.13E-04   & 1.70E-07   \\ \hline
\end{tabular}}
\caption{Estimates of the parameters of model~(\ref{eq:model_mul_comp}) when errors are stationary random variables as defined in~(\ref{eq:error_MA}) and M =  N = 25}
\label{table:13}
\end{table}

\begin{table}[p]
\centering
\resizebox{\textwidth}{!}{\begin{tabular}{||c|c||c|c|c|c|c|c|c|c||}
\hline
\multicolumn{2}{||c||}{Parameters}  & $\alpha_1$ & $\beta_1$ & $\gamma_1$ & $\delta_1$ & $\alpha_2$ & $\beta_2$ & $\gamma_2$ & $\delta_2$ \\ \hline
\multicolumn{2}{||c||}{True values} & 2.1        & 0.1       & 1.25       & 0.25       & 1.5        & 0.5       & 1.75       & 0.75       \\ \hline
$\sigma$         & & \multicolumn{8}{c||}{}                                                                                                \\ \hline
 &  & \multicolumn{8}{c||}{ALSEs} \\ \hline
 & Average & 2.1011 & 0.1000 & 1.2597 & 0.2499 & 1.5127 & 0.4997 & 1.7529 & 0.7499 \\ \hline
 & Bias & 0.0011 & -1.39E-05 & 0.0097 & -0.0001 & 0.0127 & -0.0003 & 0.0029 & -5.90E-05 \\ \hline
 & MSE & 1.19E-06 & 2.02E-10 & 9.37E-05 & 2.07E-08 & 1.19E-06 & 7.47E-08 & 8.29E-06 & 2.07E-08 \\ \hline
0.1 &  & \multicolumn{8}{c||}{LSEs} \\ \hline
 & Average & 2.1011 & 0.1000 & 1.2572 & 0.2499 & 1.5007 & 0.5000 & 1.7507 & 0.7500 \\ \hline
 & Bias & 0.0011 & -1.11E-05 & 0.0072 & -0.0001 & 0.0007 & -1.36E-05 & 0.0007 & -1.19E-05 \\ \hline
 & MSE & 1.16E-06 & 1.34E-10 & 5.19E-05 & 1.49E-08 & 6.29E-07 & 2.11E-10 & 5.36E-07 & 1.69E-10 \\ \hline
 &  &   \multicolumn{8}{c||}{}  \\ \hline
 & AVar & 2.25E-08 & 8.43E-12 & 2.25E-08 & 8.43E-12 & 7.09E-08 & 2.66E-11 & 7.09E-08 & 2.66E-11 \\ \hline
 &  & \multicolumn{8}{c||}{ALSEs} \\ \hline
 & Average & 2.1010 & 0.1000 & 1.2597 & 0.2499 & 1.5128 & 0.4997 & 1.7529 & 0.7499 \\ \hline
 & Bias & 0.0010 & -1.27E-05 & 0.0097 & -0.0001 & 0.0128 & -0.0003 & 0.0029 & -6.06E-05 \\ \hline
 & MSE & 1.52E-06 & 3.66E-10 & 9.41E-05 & 2.09E-08 & 1.52E-06 & 7.56E-08 & 1.06E-05 & 2.09E-08 \\ \hline
0.5 &  & \multicolumn{8}{c||}{LSEs} \\ \hline
 & Average & 2.1010 & 0.1000 & 1.2572 & 0.2499 & 1.5008 & 0.5000 & 1.7508 & 0.7500 \\ \hline
 & Bias & 0.0010 & -9.81E-06 & 0.0072 & -0.0001 & 0.0008 & -1.41E-05 & 0.0008 & -1.32E-05 \\ \hline
 & MSE & 1.48E-06 & 2.97E-10 & 5.23E-05 & 1.51E-08 & 2.52E-06 & 9.37E-10 & 2.28E-06 & 8.01E-10 \\ \hline
 &  &   \multicolumn{8}{c||}{}  \\ \hline
 & Avar & 5.62E-07 & 2.11E-10 & 5.62E-07 & 2.11E-10 & 1.77E-06 & 6.65E-10 & 1.77E-06 & 6.65E-10 \\ \hline
 &  & \multicolumn{8}{c||}{ALSEs} \\ \hline
 & Average & 2.1011 & 0.1000 & 1.2597 & 0.2499 & 1.5128 & 0.4997 & 1.7528 & 0.7499 \\ \hline
 & Bias & 0.0011 & -1.41E-05 & 0.0097 & -0.0001 & 0.0128 & -0.0003 & 0.0028 & -5.87E-05 \\ \hline
 & MSE & 3.16E-06 & 1.07E-09 & 9.55E-05 & 2.15E-08 & 3.16E-06 & 7.75E-08 & 1.59E-05 & 2.15E-08 \\ \hline
1 &  & \multicolumn{8}{c||}{LSEs} \\ \hline
 & Average & 2.1011 & 0.1000 & 1.2572 & 0.2499 & 1.5008 & 0.5000 & 1.7506 & 0.7500 \\ \hline
 & Bias & 0.0011 & -1.12E-05 & 0.0072 & -0.0001 & 0.0008 & -1.48E-05 & 0.0006 & -1.17E-05 \\ \hline
 & MSE & 3.07E-06 & 9.72E-10 & 5.36E-05 & 1.56E-08 & 7.29E-06 & 2.71E-09 & 7.78E-06 & 2.82E-09 \\ \hline
 &  &   \multicolumn{8}{c||}{}  \\ \hline
 & Avar & 2.25E-06 & 8.43E-10 & 2.25E-06 & 8.43E-10 & 7.09E-06 & 2.66E-09 & 7.09E-06 & 2.66E-09 \\ \hline
\end{tabular}}
\caption{Estimates of the parameters of model~(\ref{eq:model_mul_comp}) when errors are stationary random variables as in~(\ref{eq:error_MA}) and M =  N = 50}
\label{table:14}
\end{table}

\begin{table}[p]
\centering
\resizebox{\textwidth}{!}{\begin{tabular}{||c|c||c|c|c|c|c|c|c|c||}
\hline
\multicolumn{2}{||c||}{Parameters}  & $\alpha_1$ & $\beta_1$ & $\gamma_1$ & $\delta_1$ & $\alpha_2$ & $\beta_2$ & $\gamma_2$ & $\delta_2$ \\ \hline
\multicolumn{2}{||c||}{True values} & 2.1        & 0.1       & 1.25       & 0.25       & 1.5        & 0.5       & 1.75       & 0.75       \\ \hline
$\sigma$         & & \multicolumn{8}{c||}{}                                                                                                \\ \hline
    &         & \multicolumn{8}{c||}{ALSEs}                                                                  \\ \hline
    & Average & 2.0999    & 0.1000    & 1.2534   & 0.2500    & 1.5002   & 0.5000    & 1.7506   & 0.7500    \\ \hline
    & Bias    & -6.49E-05 & -4.92E-07 & 0.0034   & -4.27E-05 & 0.0002   & -2.80E-06 & 0.0006   & -6.11E-06 \\ \hline
    & MSE     & 7.52E-09  & 7.78E-13  & 1.16E-05 & 1.82E-09  & 7.52E-09 & 1.05E-11  & 3.90E-07 & 1.82E-09  \\ \hline
0.1 &         & \multicolumn{8}{c||}{LSEs}                                                                   \\ \hline
    & Average & 2.1000    & 0.1000    & 1.2528   & 0.2500    & 1.5001   & 0.5000    & 1.7500   & 0.7500    \\ \hline
    & Bias    & -2.64E-05 & 7.13E-08  & 0.0028   & -3.28E-05 & 5.65E-05 & -5.64E-07 & 2.06E-05 & -6.25E-09 \\ \hline
    & MSE     & 4.06E-09  & 5.57E-13  & 7.88E-06 & 1.07E-09  & 1.74E-08 & 2.71E-12  & 1.36E-08 & 2.13E-12  \\ \hline
    &         & \multicolumn{8}{c||}{}                                                                       \\ \hline
    & AVar    & 4.44E-09  & 7.40E-13  & 4.44E-09 & 7.40E-13  & 1.40E-08 & 2.33E-12  & 1.40E-08 & 2.33E-12  \\ \hline
    &         & \multicolumn{8}{c||}{ALSEs}                                                                  \\ \hline
    & Average & 2.0999    & 0.1000    & 1.2534   & 0.2500    & 1.5001   & 0.5000    & 1.7506   & 0.7500    \\ \hline
    & Bias    & -6.36E-05 & -5.33E-07 & 0.0034   & -4.26E-05 & 0.0001   & -2.16E-06 & 0.0006   & -6.42E-06 \\ \hline
    & MSE     & 7.47E-08  & 1.18E-11  & 1.17E-05 & 1.83E-09  & 7.47E-08 & 5.31E-11  & 7.66E-07 & 1.83E-09  \\ \hline
0.5 &         & \multicolumn{8}{c||}{LSEs}                                                                   \\ \hline
    & Average & 2.1000    & 0.1000    & 1.2528   & 0.2500    & 1.5000   & 0.5000    & 1.7500   & 0.7500    \\ \hline
    & Bias    & -2.57E-05 & 4.05E-08  & 0.0028   & -3.27E-05 & 3.31E-05 & -3.19E-07 & 3.33E-05 & -2.08E-07 \\ \hline
    & MSE     & 7.04E-08  & 1.14E-11  & 7.94E-06 & 1.08E-09  & 2.75E-07 & 4.69E-11  & 3.78E-07 & 6.20E-11  \\ \hline
    &         & \multicolumn{8}{c||}{}                                                                       \\ \hline
    & AVar    & 1.11E-07  & 1.85E-11  & 1.11E-07 & 1.85E-11  & 3.50E-07 & 5.83E-11  & 3.50E-07 & 5.83E-11  \\ \hline
    &         & \multicolumn{8}{c||}{ALSEs}                                                                  \\ \hline
    & Average & 2.0999    & 0.1000    & 1.2534   & 0.2500    & 1.5002   & 0.5000    & 1.7506   & 0.7500    \\ \hline
    & Bias    & -5.65E-05 & -7.05E-07 & 0.0034   & -4.30E-05 & 0.0002   & -2.41E-06 & 0.0006   & -6.29E-06 \\ \hline
    & MSE     & 3.03E-07  & 4.77E-11  & 1.21E-05 & 1.92E-09  & 3.03E-07 & 1.96E-10  & 1.79E-06 & 1.92E-09  \\ \hline
1   &         & \multicolumn{8}{c||}{LSEs}                                                                   \\ \hline
    & Average & 2.1000    & 0.1000    & 1.2528   & 0.2500    & 1.5001   & 0.5000    & 1.7500   & 0.7500    \\ \hline
    & Bias    & -2.21E-05 & -8.58E-08 & 0.0028   & -3.30E-05 & 7.48E-05 & -7.56E-07 & 1.94E-05 & -1.03E-07 \\ \hline
    & MSE     & 2.90E-07  & 4.60E-11  & 8.36E-06 & 1.16E-09  & 1.10E-06 & 1.87E-10  & 1.45E-06 & 2.27E-10  \\ \hline
    &         & \multicolumn{8}{c||}{}                                                                       \\ \hline
    & AVar    & 4.44E-07  & 7.40E-11  & 4.44E-07 & 7.40E-11  & 1.40E-06 & 2.33E-10  & 1.40E-06 & 2.33E-10  \\ \hline
\end{tabular}}
\caption{Estimates of the parameters of model~(\ref{eq:model_mul_comp}) when errors are stationary random variables as in~(\ref{eq:error_MA}) and M =  N = 75}
\label{table:15}
\end{table}

\begin{table}[p]
\centering
\resizebox{\textwidth}{!}{\begin{tabular}{||c|c||c|c|c|c|c|c|c|c||}
\hline
\multicolumn{2}{||c||}{Parameters}  & $\alpha_1$ & $\beta_1$ & $\gamma_1$ & $\delta_1$ & $\alpha_2$ & $\beta_2$ & $\gamma_2$ & $\delta_2$ \\ \hline
\multicolumn{2}{||c||}{True values} & 2.1        & 0.1       & 1.25       & 0.25       & 1.5        & 0.5       & 1.75       & 0.75       \\ \hline
$\sigma$         & & \multicolumn{8}{c||}{}                                                                                                \\ \hline
    &         & \multicolumn{8}{c||}{ALSEs}                                                                  \\ \hline
    & Average & 2.1006   & 0.1000    & 1.2502   & 0.2500    & 1.4991   & 0.5000    & 1.7501    & 0.7500    \\ \hline
    & Bias    & 0.0006   & -6.23E-06 & 0.0002   & -3.43E-06 & -0.0009  & 1.30E-05  & 7.61E-05  & -8.05E-08 \\ \hline
    & MSE     & 3.04E-07 & 3.90E-11  & 6.09E-08 & 1.19E-11  & 3.04E-07 & 1.71E-10  & 1.17E-08  & 1.19E-11  \\ \hline
0.1 &         & \multicolumn{8}{c||}{LSEs}                                                                   \\ \hline
    & Average & 2.0995   & 0.1000    & 1.2504   & 0.2500    & 1.5000   & 0.5000    & 1.7500    & 0.7500    \\ \hline
    & Bias    & -0.0005  & 2.66E-06  & 0.0004   & -3.55E-06 & 1.24E-05 & -9.74E-08 & -3.02E-06 & -1.67E-08 \\ \hline
    & MSE     & 2.18E-07 & 7.21E-12  & 2.00E-07 & 1.28E-11  & 4.35E-09 & 4.02E-13  & 4.26E-09  & 4.06E-13  \\ \hline
    &         & \multicolumn{8}{l||}{}                                                                       \\ \hline
    & AVar    & 1.40E-09 & 1.32E-13  & 1.40E-09 & 1.32E-13  & 4.43E-09 & 4.15E-13  & 4.43E-09  & 4.15E-13  \\ \hline
    &         & \multicolumn{8}{c||}{ALSEs}                                                                  \\ \hline
    & Average & 2.1005   & 0.1000    & 1.2502   & 0.2500    & 1.4991   & 0.5000    & 1.7501    & 0.7500    \\ \hline
    & Bias    & 0.0005   & -6.17E-06 & 0.0002   & -3.39E-06 & -0.0009  & 1.30E-05  & 8.09E-05  & -1.51E-07 \\ \hline
    & MSE     & 3.20E-07 & 4.06E-11  & 8.98E-08 & 1.47E-11  & 3.20E-07 & 1.81E-10  & 1.22E-07  & 1.47E-11  \\ \hline
0.5 &         & \multicolumn{8}{c||}{LSEs}                                                                   \\ \hline
    & Average & 2.0995   & 0.1000    & 1.2504   & 0.2500    & 1.5000   & 0.5000    & 1.7500    & 0.7500    \\ \hline
    & Bias    & -0.0005  & 2.71E-06  & 0.0004   & -3.51E-06 & 1.33E-05 & -1.38E-07 & 9.68E-07  & -7.81E-08 \\ \hline
    & MSE     & 2.48E-07 & 9.82E-12  & 2.27E-07 & 1.56E-11  & 1.19E-07 & 1.10E-11  & 1.19E-07  & 1.08E-11  \\ \hline
    &         & \multicolumn{8}{l||}{}                                                                       \\ \hline
    & AVar    & 3.51E-08 & 3.29E-12  & 3.51E-08 & 3.29E-12  & 1.11E-07 & 1.04E-11  & 1.11E-07  & 1.04E-11  \\ \hline
    &         & \multicolumn{8}{c||}{ALSEs}                                                                  \\ \hline
    & Average & 2.1006   & 0.1000    & 1.2502   & 0.2500    & 1.4990   & 0.5000    & 1.7501    & 0.7500    \\ \hline
    & Bias    & 0.0006   & -6.33E-06 & 0.0002   & -3.43E-06 & -0.0010  & 1.33E-05  & 8.50E-05  & -2.22E-07 \\ \hline
    & MSE     & 4.17E-07 & 4.96E-11  & 1.86E-07 & 2.43E-11  & 4.17E-07 & 2.17E-10  & 4.81E-07  & 2.43E-11  \\ \hline
1   &         & \multicolumn{8}{c||}{LSEs}                                                                   \\ \hline
    & Average & 2.0996   & 0.1000    & 1.2504   & 0.2500    & 1.5000   & 0.5000    & 1.7500    & 0.7500    \\ \hline
    & Bias    & -0.0004  & 2.55E-06  & 0.0004   & -3.54E-06 & 2.30E-06 & 9.41E-08  & 2.39E-06  & -1.29E-07 \\ \hline
    & MSE     & 2.97E-07 & 1.57E-11  & 3.27E-07 & 2.51E-11  & 4.52E-07 & 4.21E-11  & 4.90E-07  & 4.52E-11  \\ \hline
    &         & \multicolumn{8}{c||}{}                                                                       \\ \hline
    & AVar    & 1.40E-07 & 1.32E-11  & 1.40E-07 & 1.32E-11  & 4.43E-07 & 4.15E-11  & 4.43E-07  & 4.15E-11  \\ \hline
\end{tabular}}
\caption{Estimates of the parameters of model~(\ref{eq:model_mul_comp}) when errors are stationary random variables as in~(\ref{eq:error_MA}) and M =  N = 100}
\label{table:16}
\end{table}

\section{Data Analysis}\label{sec:8}
We perform analysis of a simulated data set to exemplify how we can extract regular gray-scale texture from the one that is contaminated with noise. The data $y(m, n)$ is generated using~(\ref{eq:model_1}) with the following true parameter values:
$$A^0 = 6,\ B^0 = 6,\ \alpha^0 = 2.75,\ \beta^0 = 0.05,\ \gamma^0 = 2.5\ \textmd{ and } \delta^0 = 0.075$$
and the error random variables $\{X(m, n)\}$ are generated as follows:
\begin{equation*}
X(m, n) = \epsilon(m, n) + \rho_1 \epsilon(m-1,n) + \rho_2 \epsilon(m, n-1) + \rho_3 \epsilon(m-1, n-1),
\end{equation*}
where $\epsilon(m, n)$ are Gaussian random variables with mean 0 and variance $\sigma^2 = 2$ and $\rho_1 = 0.5$, $\rho_2 = 0.4$ and $\rho_3 = 0.3$. We generate $y(m, n)$ for $M = N = 100$. Figure~\ref{fig:true_signal} displays the true signal generated with the above mentioned parameters and Figure~\ref{fig:noisy_signal} displays the noisy signal that is the true signal along with the additive error $X(m, n)$ defined above.
\graphicspath{C:/Users/DELL-PC/Dropbox/Paper2}
\begin{figure}[H]
\centering
\begin{minipage}{0.5\textwidth}
\includegraphics[scale=0.2]{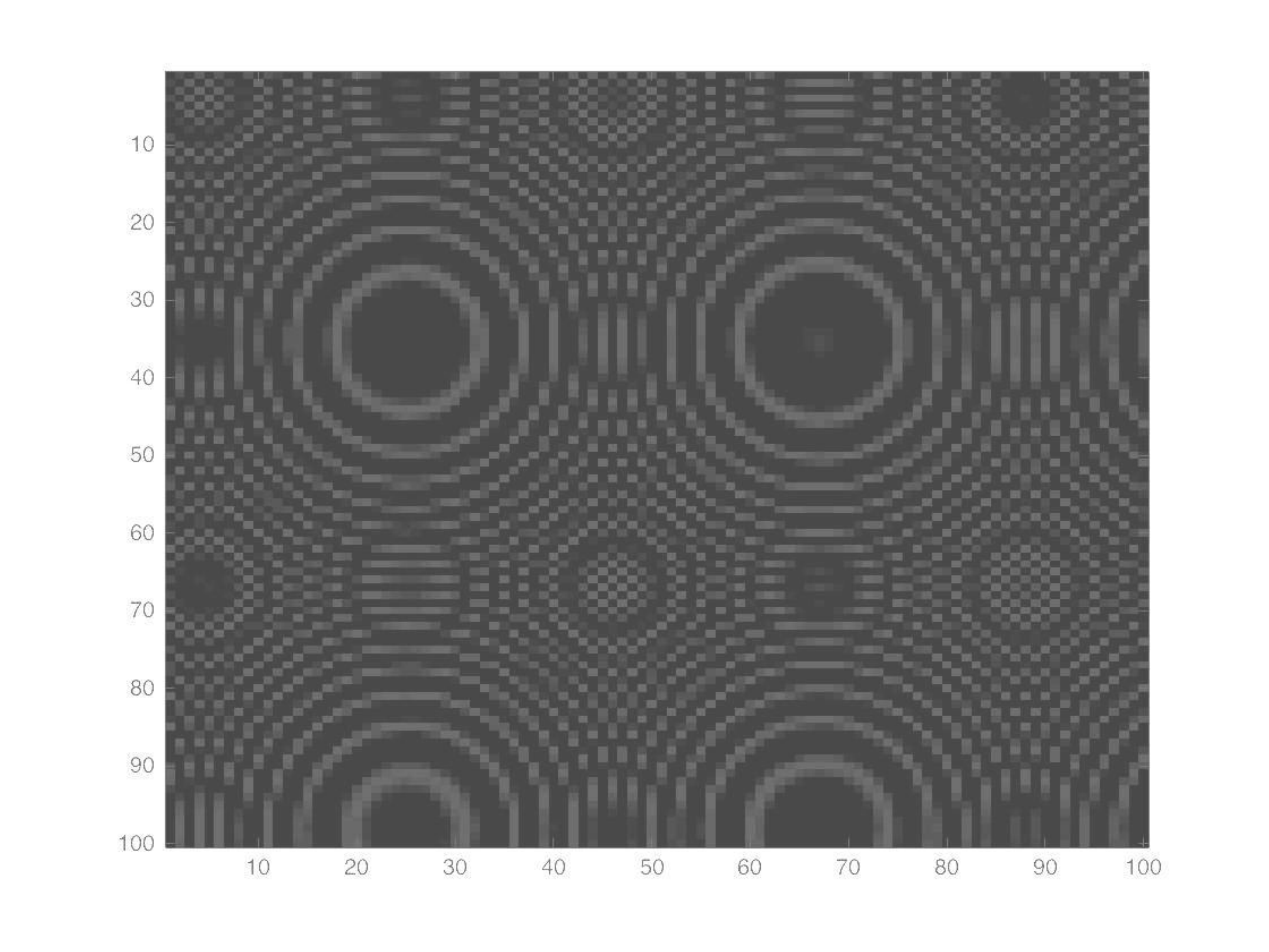}
\caption{True Signal}
\label{fig:true_signal}
\end{minipage}%
\begin{minipage}{0.5\textwidth}
\includegraphics[scale=0.2]{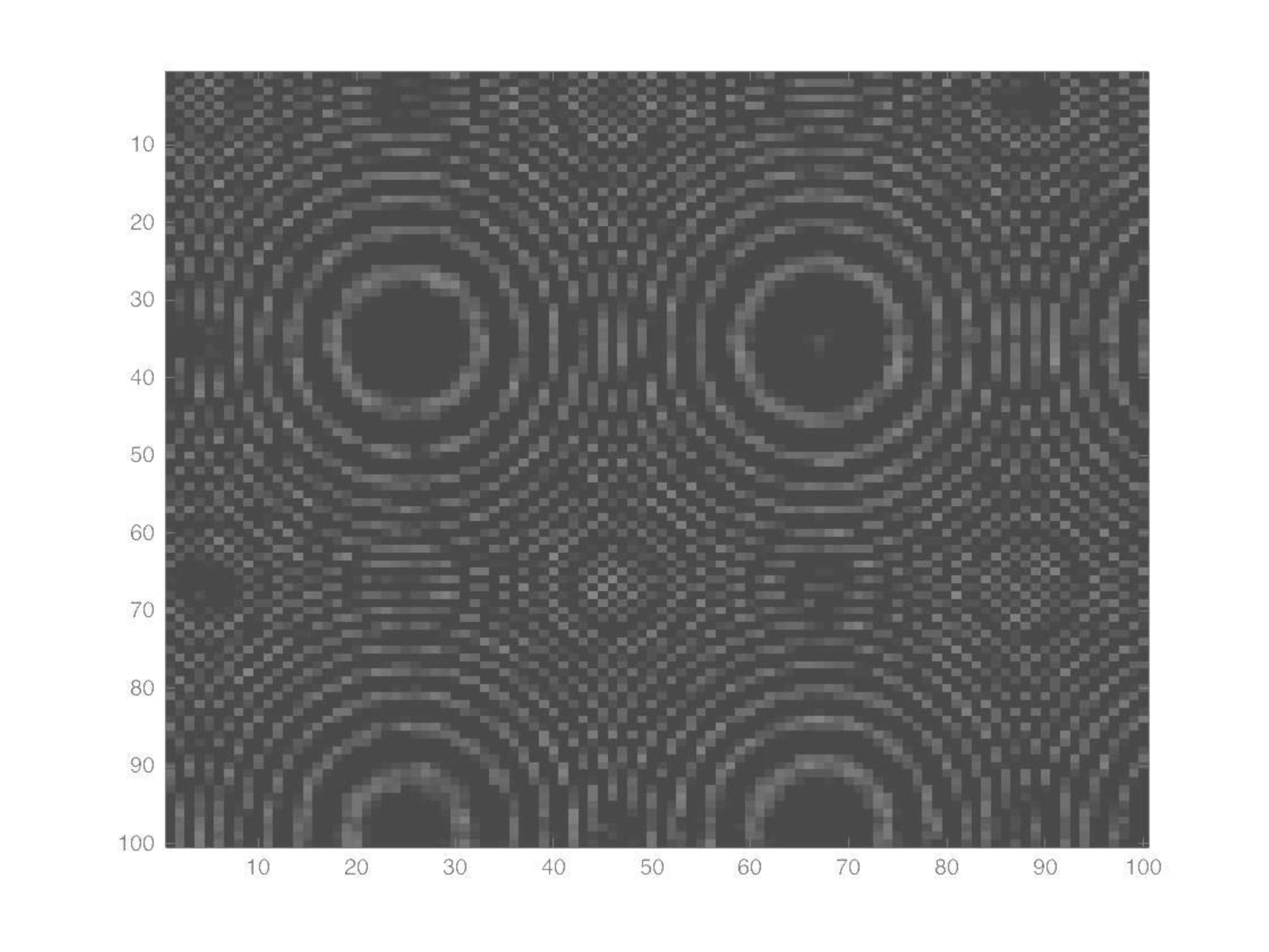}
\caption{Noisy Signal}
\label{fig:noisy_signal}
\end{minipage}
\end{figure}
\justify
Using the generated data matrix, we now fit model~(\ref{eq:model_1}) using both least squares estimation method and approximate least squares estimation method. 
 Following are the values of the estimates that we attain:
\begin{align*}\begin{split}
\textmd{ LSEs: } \hat{A} = 5.909047,\ \hat{B} = 6.073225,\ \hat{\alpha} = 2.750070,\  \hat{\beta}= 0.049998,\ \hat{\gamma} = 2.500897,\ \hat{\delta} = 0.074992 \\
\textmd{ALSEs: } \tilde{A} = 5.817259,\ \tilde{B} = 6.152546,\ \tilde{\alpha} = 2.751267,\ \tilde{\beta}= 0.049986,\ \tilde{\gamma} = 2.500725,\ \tilde{\delta} = 0.074993
\end{split} \end{align*}

\begin{figure}[H]
\centering
\begin{minipage}{0.5\textwidth}
\includegraphics[scale=0.2]{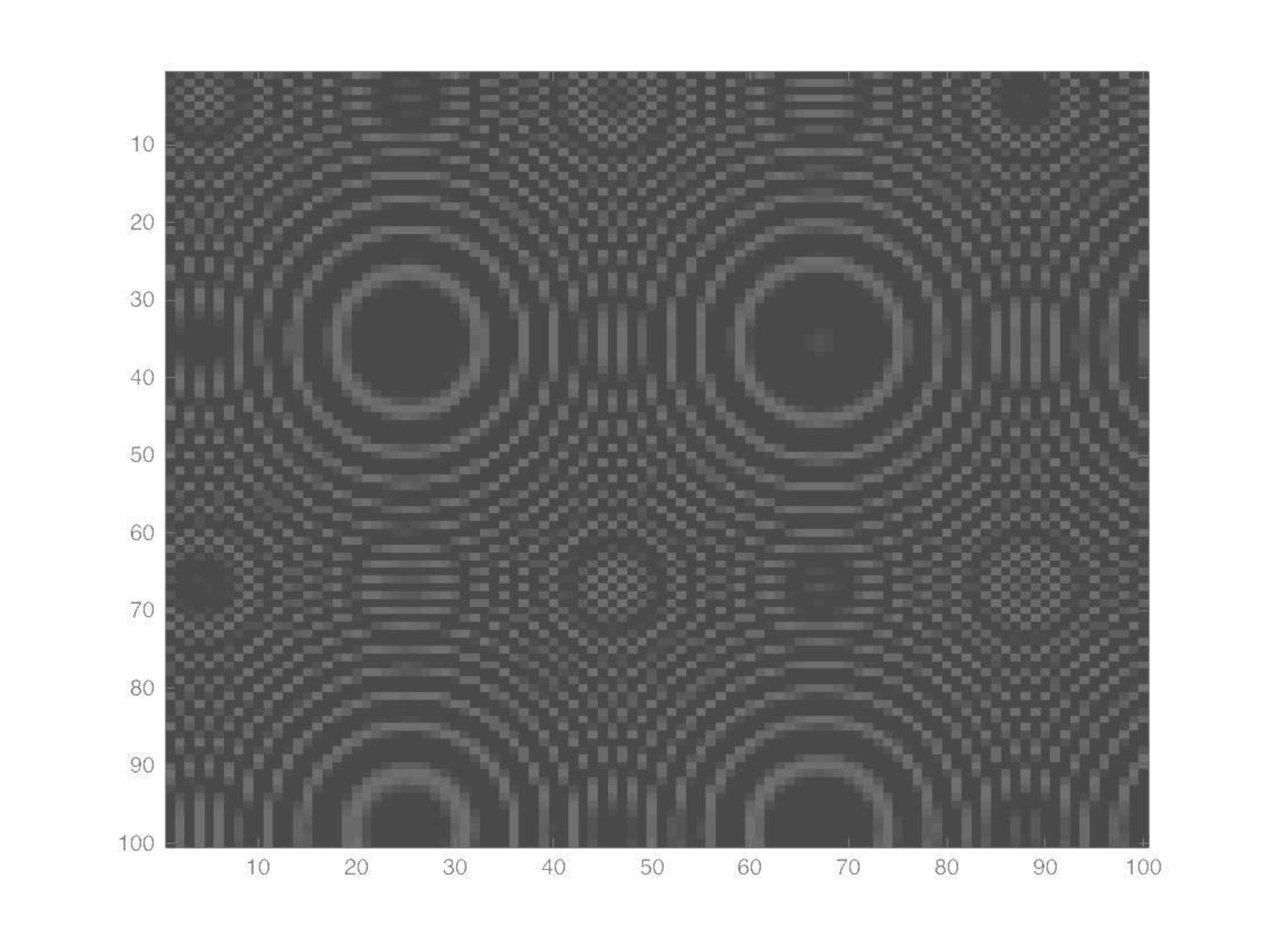}
\caption{Estimated Signal using LSEs}
\label{fig:estimated_signal_LSE}
\end{minipage}%
\begin{minipage}{0.5\textwidth}
\includegraphics[scale=0.2]{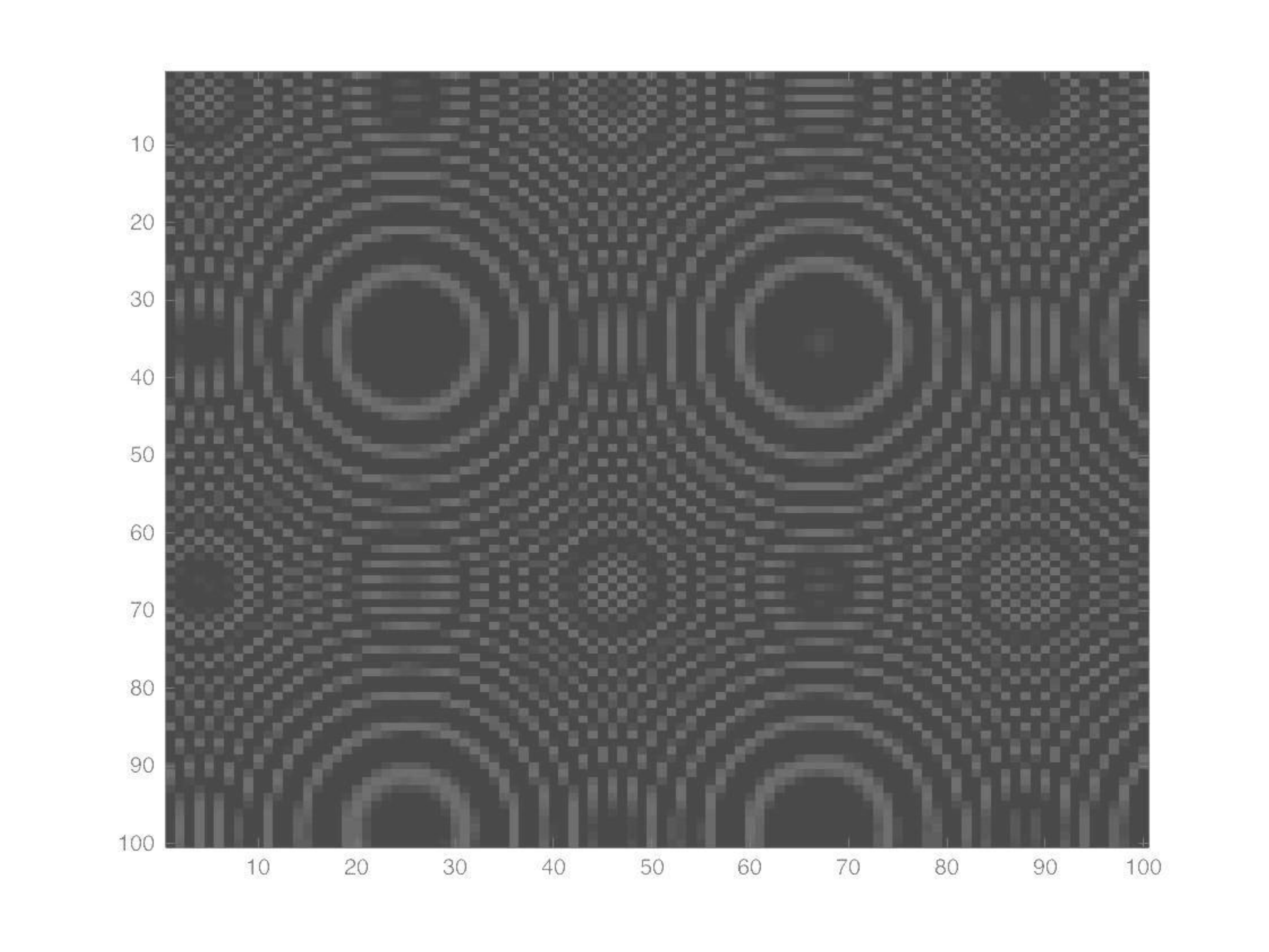}
\caption{Estimated Signal using ALSEs}
\label{fig:estimated_signal_ALSE}
\end{minipage}
\end{figure}
\justify
We plot the estimated signals using the above obtained estimates. The plots of the signals estimated using LSEs and ALSEs are given in Figures~\ref{fig:estimated_signal_LSE} and~\ref{fig:estimated_signal_ALSE} respectively.
Thus we may conclude that the estimated signal plots both using LSEs and ALSEs are well matched with the true signal plot, as is evident from the figures above.

\section{Conclusion}\label{sec:9}
In this paper, we propose approximate least squares estimators (ALSEs) to estimate the unknown parameters of a one component 2-D chirp model. We study their asymptotic properties. We show that they are strongly consistent and asymptotically, normally distributed and equivalent to the LSEs. The consistency of the ALSEs of the linear parameters is obtained under slightly weaker conditions than that obtained for the LSEs of the linear parameters, as we need not bound the parameter space in the former case. Also, the rate of convergence of the linear parameters is $M^{-1/2}N^{-1/2}$, that of frequencies $\alpha$ and $\gamma$ are $M^{-3/2}N^{-1/2}$ and $M^{-1/2}N^{-3/2}$ respectively and  that of frequency rates $\beta$ and $\delta$ are $M^{-5/2}N^{-1/2}$ and $M^{-1/2}N^{-5/2}$ respectively, same as that of corresponding LSEs.  Through simulation studies as well, we deduce that the estimators are consistent and asymptotically equivalent to the LSEs. We also propose sequential procedure to obtain the ALSEs of a multiple component 2-D chirp model, with the number of components to be known and study their asymptotic properties. We see that the results obtained for the one component model can be extended to the generalised model, that is the multiple component model.  
\appendix
\section*{Appendix A}\label{appendix:A}
\justify
The following lemmas are required to prove Theorem~\ref{theorem:1}.
\begin{lemma}\label{lemma:6}
Consider the set $S_c^{\boldsymbol{\vartheta}^0}$ = $\{\boldsymbol{\vartheta} : |\boldsymbol{\vartheta} - \boldsymbol{\vartheta}^0| \geqslant 4c\}$, if for any c $>$ 0, \\
\begin{equation}\label{eq:lemma6}
\limsup \sup\limits_{\boldsymbol{\vartheta} \in S_c^{\vartheta^0}} \frac{1}{M N}(I(\boldsymbol{\vartheta}) - I(\boldsymbol{\vartheta}^0)) < 0 \ a.s.
\end{equation}
then, $\tilde{\boldsymbol{\vartheta}}$ $\rightarrow$ $\boldsymbol{\vartheta}^0$ almost surely as min$\{M,N\} \rightarrow \infty$.
\end{lemma}
\begin{proof}
Let us denote $\boldsymbol{\vartheta}^{0}$ = ($\alpha^0$, $\beta^0$, $\gamma^0$, $\delta^0$) and $\tilde{\boldsymbol{\vartheta}}$ = ($\tilde{\alpha}$, $\tilde{\beta}$, $\tilde{\gamma}$, $\tilde{\delta}$) by $\tilde{\boldsymbol{\vartheta}}_{MN}$  to assert that it depends on $M$ and $N$.
 Suppose~(\ref{eq:lemma6}) is true and $\tilde{\boldsymbol{\vartheta}}_{MN}$ $\nrightarrow$ $\boldsymbol{\vartheta}^{0}$ almost surely as min$\{M, N\}$ $\rightarrow$ $\infty$, then there exists a $c > 0$ and a subsequence \{$M_k,$ $N_k$\} of $\{M, N\}$ such that $\tilde{\boldsymbol{\vartheta}}_{M_kN_k}$ $\in$ $S_c^{\vartheta^0}$ for all $k = 1, 2, \cdots$.\\
Since $\tilde{\boldsymbol{\vartheta}}_{M_kN_k}$ is the ALSE of $\boldsymbol{\vartheta}^{0}$ when $M = M_k$ and $N = N_k$, 
$$\limsup \sup\limits_{\boldsymbol{\vartheta} \in S_c^{\vartheta^0}} \frac{1}{M_kN_k}\bigg(I_{M_kN_k}(\boldsymbol{\vartheta}) - I_{M_kN_k}(\boldsymbol{\vartheta}^{0})\bigg) \geqslant 0$$
This contradicts~(\ref{eq:lemma6}). Hence $\tilde{\alpha}$ $\xrightarrow{a.s.}$ $\alpha^0$, $\tilde{\beta}$ $\xrightarrow{a.s.}$ $\beta^0$, $\tilde{\gamma}$ $\xrightarrow{a.s.}$ $\gamma^0$ and $\tilde{\delta}$ $\xrightarrow{a.s.}$ $\delta^0$.\\
\end{proof}

\begin{lemma}\label{lemma:7}
If assumptions 1 and 2 are satisfied then:
\begin{equation*}\begin{split}
M(\tilde{\alpha} - \alpha^0) \xrightarrow{a.s.} 0,\ M^2(\tilde{\beta} - \beta^0) \xrightarrow{a.s.} 0,\ N(\tilde{\gamma} - \gamma^0) \xrightarrow{a.s.} 0 \textmd{ and } N^2(\tilde{\delta} - \delta^0) \xrightarrow{a.s.} 0 
\end{split}\end{equation*}
\end{lemma}
\begin{proof}
Let $I'(\boldsymbol{\vartheta})$ be 1 $\times$ 4 first derivative vector and $I''(\boldsymbol{\vartheta})$ be 4 $\times$ 4 second derivative matrix of $I(\boldsymbol{\vartheta})$. Using multivariate Taylor series expansion of $I(\boldsymbol{\vartheta})$ around $\boldsymbol{\vartheta}^0$, we have:
\begin{align*}
I'(\tilde{\boldsymbol{\vartheta}})  -  I'(\boldsymbol{\vartheta}^0) = (\tilde{\boldsymbol{\vartheta}} - \boldsymbol{\vartheta}^0)I''(\bar{\boldsymbol{\vartheta}})
\end{align*}
Since $\tilde{\boldsymbol{\vartheta}}$ is the ALSE of $\boldsymbol{\vartheta}^0$, $I'(\tilde{\boldsymbol{\vartheta}})$ = 0 
\begin{equation}\label{eq:non_lin}
(\tilde{\boldsymbol{\vartheta}} - \boldsymbol{\vartheta}^0)(\sqrt{MN}\mathbf{D}_1)^{-1} = -\frac{1}{\sqrt{MN}}I'(\boldsymbol{\vartheta}^0)\mathbf{D}_1[\mathbf{D}_1I''(\bar{\boldsymbol{\vartheta}})\mathbf{D}_1]^{-1}
\end{equation}
where, $\mathbf{D}_1$ = $\textnormal{diag}(M^{-3/2}N^{-1/2}, M^{-5/2}N^{-1/2}, M^{-1/2}N^{-3/2}, M^{-1/2}N^{-5/2}).$ To show that the left hand side of equation~(\ref{eq:non_lin}) goes to 0 as min$\{M, N\}$ $\rightarrow$ $\infty$, we first consider the vector  $\frac{1}{\sqrt{MN}}I'(\boldsymbol{\vartheta}^0)$
\begin{equation}\begin{split}\label{eq:I_vector}
&{ = \frac{1}{\sqrt{MN}}\left[\begin{array}{cccc}\frac{\partial{I(\boldsymbol{\vartheta}^0)}}{\partial{\alpha}} & \frac{\partial{I(\boldsymbol{\vartheta}^0)}}{\partial{\beta}} & \frac{\partial{I(\boldsymbol{\vartheta}^0)}}{\partial{\gamma}} & \frac{\partial{I(\boldsymbol{\vartheta}^0)}}{\partial{\delta}}\end{array}\right] \left[\begin{array}{cccc}\frac{1}{M^{\frac{3}{2}}N^{\frac{1}{2}}} & 0 & 0 & 0 \\0 & \frac{1}{M^{\frac{5}{2}}N^{\frac{1}{2}}}  & 0 & 0 \\0 & 0 & \frac{1}{M^{\frac{1}{2}}N^{\frac{3}{2}}}  & 0 \\0 & 0 & 0 & \frac{1}{M^{\frac{1}{2}}N^{\frac{5}{2}}} \end{array}\right]} \\
&{= \left[\begin{array}{cccc}\frac{1}{M^2N}\frac{\partial{I(\boldsymbol{\vartheta}^0)}}{\partial{\alpha}} & \frac{1}{M^3N}\frac{\partial{I(\boldsymbol{\vartheta}^0)}}{\partial{\beta}} & \frac{1}{MN^2}\frac{\partial{I(\boldsymbol{\vartheta}^0)}}{\partial{\gamma}} & \frac{1}{MN^3}\frac{\partial{I(\boldsymbol{\vartheta}^0)}}{\partial{\delta}}\end{array}\right] } \\
\end{split}\end{equation}
From~(\ref{eq:per}), we can write $I(\boldsymbol{\vartheta})$ as:
\begin{equation*}\begin{split}
& = \frac{2}{MN} \bigg(\sum_{n = 1}^{N}\sum_{m = 1}^{M} y(m, n)\cos(\alpha^0 m + \beta^0 m^2 + \gamma^0 n + \delta^0 n^2)\bigg)^2 + \\
& \quad \hspace{2mm} \frac{2}{MN} \bigg(\sum_{n = 1}^{N}\sum_{m = 1}^{M} y(m, n)\sin(\alpha^0 m + \beta^0 m^2 + \gamma^0 n + \delta^0 n^2)\bigg)^2.
\end{split}\end{equation*}
Thus,
\vspace{-3mm}
\begin{equation*}\begin{split}
\frac{1}{M^2N} \frac{\partial{I(\boldsymbol{\vartheta}^0)}}{\partial{\alpha}}&  =  \frac{4}{M^3N^2} \bigg(\sum_{n = 1}^{N}\sum_{m = 1}^{M} y(m, n)\cos(\alpha^0 m + \beta^0 m^2 + \gamma^0 n + \delta^0 n^2)\bigg) \\
& \quad \quad \bigg(-\sum_{n = 1}^{N}\sum_{m = 1}^{M}m\ y(m, n)\sin(\alpha^0 m + \beta^0 m^2 + \gamma^0 n + \delta^0 n^2)\bigg) \\
& + \frac{4}{M^3N^2} \bigg(\sum_{n = 1}^{N}\sum_{m = 1}^{M} y(m, n)\sin(\alpha^0 m + \beta^0 m^2 + \gamma^0 n + \delta^0 n^2)\bigg) \\
&  \quad \quad \quad \bigg(\sum_{n = 1}^{N}\sum_{m = 1}^{M}m\ y(m, n)\cos(\alpha^0 m + \beta^0 m^2 + \gamma^0 n + \delta^0 n^2)\bigg).
\end{split}\end{equation*}
Now using equation~(\ref{eq:model_1}), taking limit as min$\{M, N\}$ $\rightarrow$ $\infty$, and then using Lemma ~\ref{lemma:1}, parts \textit{(c)-(g)}, we have:
\begin{equation*}
\frac{1}{M^2N} \frac{\partial{I(\boldsymbol{\vartheta}^0)}}{\partial{\alpha}}  \xrightarrow{a.s.} 0.  
\end{equation*}
On similar lines, one can show that rest of the elements of the above vector~(\ref{eq:I_vector}) tend to 0 as min$\{M, N\}$ $\rightarrow$ $\infty$.
Thus, we have:
\begin{equation}\label{eq:I_part_1}
\lim_{n \rightarrow \infty}\frac{1}{\sqrt{MN}}I'(\boldsymbol{\vartheta}^0) = 0.
\end{equation}
\vspace{-2mm}
\noindent Now we consider the second derivative matrix: $\mathbf{D}_1I''(\bar{\boldsymbol{\vartheta}})\mathbf{D}_1.$ Since $I''(\boldsymbol{\vartheta})$ is a continuous function of $\boldsymbol{\vartheta}$, 
\begin{equation*}
\lim\limits_{n \rightarrow \infty} \mathbf{D}_1I''(\bar{\boldsymbol{\vartheta}})\mathbf{D}_1 = \lim\limits_{n \rightarrow \infty} \mathbf{D}_1I''(\boldsymbol{\vartheta}^0)\mathbf{D}_1,
\end{equation*}
where,
\begin{equation*}\begin{split}
\mathbf{D}_1I''(\boldsymbol{\vartheta}^0)\mathbf{D}_1 & {\large =\left[\begin{array}{cccc} \frac{1}{M^3N}\frac{\partial^2I(\boldsymbol{\vartheta}^0)}{\partial\alpha^2} &  \frac{1}{M^4N} \frac{\partial^2I(\boldsymbol{\vartheta}^0)}{\partial\alpha\partial\beta} &  \frac{1}{M^2N^2} \frac{\partial^2I(\boldsymbol{\vartheta}^0)}{\partial\alpha\partial\gamma} &  \frac{1}{M^2N^3} \frac{\partial^2I(\boldsymbol{\vartheta}^0)}{\partial\alpha\partial\delta} \\  \frac{1}{M^4N}\frac{\partial^2I(\boldsymbol{\vartheta}^0)}{\partial\beta\partial\alpha} &  \frac{1}{M^5N} \frac{\partial^2I(\boldsymbol{\vartheta}^0)}{\partial\beta^2} &  \frac{1}{M^3N^2} \frac{\partial^2I(\boldsymbol{\vartheta}^0)}{\partial\beta\partial\gamma} & \frac{1}{M^3N^3} \frac{\partial^2I(\boldsymbol{\vartheta}^0)}{\partial\beta\partial\delta} \\ \frac{1}{M^2N^2}\frac{\partial^2I(\boldsymbol{\vartheta}^0)}{\partial\gamma\partial\alpha} & \frac{1}{M^3N^2} \frac{\partial^2I(\boldsymbol{\vartheta}^0)}{\partial\gamma\partial\beta} & \frac{1}{MN^3} \frac{\partial^2I(\boldsymbol{\vartheta}^0)}{\partial\gamma^2} &  \frac{1}{MN^4} \frac{\partial^2I(\boldsymbol{\vartheta}^0)}{\partial\gamma\partial\delta} \\ \frac{1}{M^2N^3}\frac{\partial^2I(\boldsymbol{\vartheta}^0)}{\partial\delta\partial\alpha} & \frac{1}{M^3N^3} \frac{\partial^2I(\boldsymbol{\vartheta}^0)}{\partial\delta\partial\beta} & \frac{1}{MN^4} \frac{\partial^2I(\boldsymbol{\vartheta}^0)}{\partial\delta\partial\gamma} &  \frac{1}{MN^5}\frac{\partial^2I(\boldsymbol{\vartheta}^0)}{\partial\delta^2}\end{array}\right]}.
\end{split}\end{equation*}
Using Lemma~\ref{lemma:1}, parts \textit{(c)} - \textit{(g)} provided in section~\ref{sec:2and3 }, we obtain the following:
\vspace{-2mm}
\begin{equation}\label{eq:I_part_2}
\lim\limits_{\textnormal{min}\{M, N\} \rightarrow \infty} \mathbf{D}_1I''(\boldsymbol{\vartheta}^0)\mathbf{D}_1 \overset{\mathrm{a.s.}}{=} -\mathbf{S}  \\
\end{equation}
where,
\begin{equation*}
\mathbf{S} = {\large \left[\begin{array}{cccc}\frac{{A^0}^2 + {B^0}^2}{12} & \frac{{A^0}^2 + {B^0}^2}{12} & 0 & 0 \\\frac{{A^0}^2 + {B^0}^2}{12} & \frac{4({A^0}^2 + {B^0}^2)}{45} & 0 & 0 \\0 & 0 & \frac{{A^0}^2 + {B^0}^2}{12} & \frac{{A^0}^2 + {B^0}^2}{12} \\0 & 0 & \frac{{A^0}^2 + {B^0}^2}{12} & \frac{4({A^0}^2 + {B^0}^2)}{45} \end{array}\right]} > 0.
\end{equation*}
Here, a matrix $A > 0$, means that it is a positive definite matrix.
From~(\ref{eq:I_part_1}) and~(\ref{eq:I_part_2}), we get the desired result.\\
\end{proof}
\justify
\textit{Proof of Theorem 1:}
We first prove the consistency of the non-linear parameters.\\
Consider the difference:
\begin{flalign*}
\frac{1}{MN}\Bigg[I(\boldsymbol{\vartheta}) - I(\boldsymbol{\vartheta}^0)\Bigg] & = \frac{2}{(MN)^2} \bigg|\sum_{m = 1}^{M}\sum_{n = 1}^{N}y(m ,n)e^{-i(\alpha m + \beta m^2 + \gamma n + \delta n^2)}\bigg|^2 - & \\
& \quad \hspace{2mm} \frac{2}{(MN)^2} \bigg|\sum_{m = 1}^{M}\sum_{n = 1}^{N}y(m ,n)e^{-i(\alpha^0 m + \beta^0 m^2 + \gamma^0 n + \delta^0 n^2)}\bigg|^2&
\end{flalign*}
\justify
\vspace{-5mm}
Consider the set $S_c^{\vartheta^0}$ defined in~\ref{lemma:6}. We split this set into four sets and thus it can be rewritten as $S_c^{(\vartheta^0,\alpha)} \cup S_c^{(\vartheta^0,\beta)} \cup S_c^{(\vartheta^0,\gamma)} \cup S_c^{(\vartheta^0,\delta)}.$
\justify
\vspace{-5mm}
Here, $S_c^{(\vartheta^0,\alpha)} = \{\alpha: |\alpha - \alpha^0| > c\},$
$S_c^{(\vartheta^0,\beta)} = \{ \beta: |\beta - \beta^0| > c\},$
$S_c^{(\vartheta^0,\gamma)} = \{ \gamma: |\gamma - \gamma^0| > c\}$ and
$S_c^{(\vartheta^0,\delta)} = \{ \delta: |\delta - \delta^0| > c\}.$
\vspace{-2mm}
\begin{equation*}\begin{split}
& \limsup \sup\limits_{\boldsymbol{\vartheta} \in S_c^{(\vartheta^0,\alpha)}}\frac{1}{MN}\Bigg[I(\boldsymbol{\vartheta}) - I(\boldsymbol{\vartheta}^0)\Bigg] \\
& = - 2\Bigg[ \lim\limits_{\textnormal{min}\{M, N\} \rightarrow \infty}\Bigg\{\frac{1}{MN}\sum_{m = 1}^{M}\sum_{n = 1}^{N}A^0\cos^2(\alpha^0 m + \beta^0 m^2 + \gamma^0 n + \delta^0 n^2 )\Bigg\}^2 + \\
& \hspace{35mm} \Bigg\{\frac{1}{MN}\sum_{m = 1}^{M}\sum_{n = 1}^{N}B^0\sin^2(\alpha^0 m + \beta^0 m^2 + \gamma^0 n + \delta^0 n^2)\Bigg\}^2\Bigg]\\
& = -\frac{1}{2}({A^0}^2 + {B^0}^2) < 0 \ a.s.
\end{split}\end{equation*}
using Lemma~\ref{lemma:1}, parts \textit{(e)} and \textit{(f)}.
Similarly, for all other sets  $S_c^{(\vartheta^0,\beta)}$, $S_c^{(\vartheta^0,\gamma)}$ and $S_c^{(\vartheta^0,\delta)}$, this can be shown.\\
Hence combining, we have:
$\limsup \sup\limits_{\boldsymbol{\vartheta} \in S_c^{\vartheta^0}}\frac{1}{MN}\Bigg[I(\boldsymbol{\vartheta}) - I(\boldsymbol{\vartheta}^0)\Bigg] <$ 0 a.s.
Thus, using Lemma~\ref{lemma:6}, we get the desired result.\\ \qed
\vspace{-5mm}
\section*{Appendix B}\label{appendix:B}
\textit{Proof of Theorem 2:} Consider the following:
\vspace{-5mm}
\begin{flalign*}
& \frac{1}{MN}Q_{MN}(\boldsymbol{\theta}) = \frac{1}{MN}\sum_{m = 1}^{M}\sum_{n = 1}^{N}\bigg(y(m, n) - A\cos(\boldsymbol{\vartheta}^{T}u(m,n)) - B \sin(\boldsymbol{\vartheta}^{T}u(m,n))\bigg)^2 &\\
& = C - \frac {1}{MN}J_{MN}(\boldsymbol{\theta}) + o(1)&\\
&\textmd{Here, } C =  \frac{1}{MN}\sum_{m = 1}^{M}\sum_{n = 1}^{N}y(m, n)^2 \textmd{ and, }& \\
&\frac {1}{MN}J_{MN}(\boldsymbol{\theta}) = \frac{2}{MN}\sum_{m = 1}^{M}\sum_{n = 1}^{N} y(m, n)\bigg(A\cos(\alpha m + \beta m^2 + \gamma n + \delta n^2) + B \sin(\alpha m + \beta m^2 + \gamma n + \delta n^2)\bigg)&\\
& \hspace{25mm}- \frac{A^2 + B^2}{2}.&
\end{flalign*}
Now we compute the first derivative of $\frac {1}{MN}J_{MN}(\boldsymbol{\theta})$ and $\frac {1}{MN}Q_{MN}(\boldsymbol{\theta})$  at $\boldsymbol{\theta}$ = $\boldsymbol{\theta}^0$ and using Lemma~\ref{lemma:1}, parts \textit{(c) - (g)}, we obtain the following relation between them:
\begin{equation}\label{eq:relation_Q_J}
\mathbf{Q}'_{MN}(\boldsymbol{\theta}^0)\mathbf{D} = - \mathbf{J}'_{MN}(\boldsymbol{\theta}^0)\mathbf{D} + 
\begin{pmatrix} 
o(M^{\frac{1}{2}}N^{\frac{1}{2}})\\
o(M^{\frac{1}{2}}N^{\frac{1}{2}}) \\
o(M^{\frac{3}{2}}N^{\frac{1}{2}}) \\
o(M^{\frac{5}{2}}N^{\frac{1}{2}}) \\
o(M^{\frac{1}{2}}N^{\frac{3}{2}})\\
o(M^{\frac{1}{2}}N^{\frac{5}{2}})
\end{pmatrix}^T \mathbf{D}. \\
\end{equation}
Since the second expression of equation(~\ref{eq:relation_Q_J}) goes to 0, as $\textnormal{min}\{M, N\} \rightarrow \infty$, we have:
\begin{equation}\label{eq:J_part_1}
\lim\limits_{\textnormal{min}\{M, N\} \rightarrow \infty} \mathbf{Q}'_{MN}(\boldsymbol{\theta}^0)\mathbf{D} =  \lim\limits_{\textnormal{min}\{M, N\} \rightarrow \infty} -\mathbf{J}'_{MN}(\boldsymbol{\theta}^0)\mathbf{D}.
\end{equation}
It can be easily seen that, at $\tilde{A}$ = $\hat{A}(\alpha , \beta, \gamma, \delta)$ and $\tilde{B}$ = $\hat{B}(\alpha , \beta, \gamma, \delta)$, 
\begin{equation*}
J_{MN}(\tilde{A}, \tilde{B}, \alpha, \beta, \gamma, \delta) = \frac{1}{MN} I(\alpha , \beta, \gamma, \delta),
\end{equation*} 
\justify
\vspace{-3mm}
where $I(\alpha , \beta, \gamma, \delta)$ is as defined in~(\ref{eq:per}).
Hence the estimator of $\boldsymbol{\theta}^0$ which maximizes $J_{MN}(\boldsymbol{\theta})$ is equivalent to $\tilde{\boldsymbol{\theta}}$, the ALSE of $\boldsymbol{\theta}^0$. Thus, the ALSE $\tilde{\boldsymbol{\theta}}$ in terms of $J_{MN}(\boldsymbol{\theta})$ can be written as the following, using Taylor series expansion:
\begin{equation}\label{eq:J_taylor}
 (\tilde{\boldsymbol{\theta}} - \boldsymbol{\theta}^0)\mathbf{D}^{-1} = -[\mathbf{J}_{MN}'(\boldsymbol{\theta}^0)\mathbf{D}][\mathbf{D}\mathbf{J}_{MN}''(\bar{\boldsymbol{\theta}})\mathbf{D}]^{-1}.
\end{equation}
Note that, $\lim\limits_{\textnormal{min}\{M, N\} \rightarrow \infty}[\mathbf{D}\mathbf{J}_{MN}''(\bar{\boldsymbol{\theta}})\mathbf{D}] = \lim\limits_{\textnormal{min}\{M, N\} \rightarrow \infty}[\mathbf{D}\mathbf{J}_{MN}''(\boldsymbol{\theta}^0)\mathbf{D}].$
Now comparing the corresponding elements of the second derivative matrices $\mathbf{D}\mathbf{J}_{MN}''(\boldsymbol{\theta}^0)\mathbf{D}$ and $\mathbf{D}\mathbf{Q}_{MN}''(\boldsymbol{\theta}^0)\mathbf{D}$ and using Lemma~\ref{lemma:1}, parts \textit{(c) - (f)}, on each of the derivatives as done for the first derivative vectors above, we obtain the following relation: 
\begin{equation}\label{eq:J_part_2}
\lim_{\textnormal{min}\{M, N\} \rightarrow \infty} \mathbf{D}\mathbf{J}_{MN}''(\boldsymbol{\theta}^0)\mathbf{D} = - \lim_{\textnormal{min}\{M, N\} \rightarrow \infty}\mathbf{D}\mathbf{Q}_{MN}''(\boldsymbol{\theta}^0)\mathbf{D} = -2\boldsymbol{\Sigma}
\end{equation}
where, $$\boldsymbol{\Sigma} = \left[\begin{array}{cccccc}1/2 & 0 & \frac{B^0}{4} & \frac{B^0}{6} & \frac{B^0}{4} & \frac{B^0}{6} \\0 & 1/2 & -\frac{A^0}{4} & -\frac{A^0}{6} & -\frac{A^0}{4} & -\frac{A^0}{6} \\\frac{B^0}{4} & -\frac{A^0}{4} & \frac{{A^0}^2 + {B^0}^2}{6} & \frac{{A^0}^2 + {B^0}^2}{8} & \frac{{A^0}^2 + {B^0}^2}{8} & \frac{{A^0}^2 + {B^0}^2}{12} \\ \frac{B^0}{6} & -\frac{A^0}{6} & \frac{{A^0}^2 + {B^0}^2}{8} & \frac{{A^0}^2 + {B^0}^2}{10} & \frac{{A^0}^2 + {B^0}^2}{12} & \frac{{A^0}^2 + {B^0}^2}{18} \\
\frac{B^0}{4} & -\frac{A^0}{4} & \frac{{A^0}^2 + {B^0}^2}{8} & \frac{{A^0}^2 + {B^0}^2}{12} & \frac{{A^0}^2 + {B^0}^2}{6} & \frac{{A^0}^2 + {B^0}^2}{8} \\ \frac{B^0}{6} & -\frac{A^0}{6} & \frac{{A^0}^2 + {B^0}^2}{12} & \frac{{A^0}^2 + {B^0}^2}{18} & \frac{{A^0}^2 + {B^0}^2}{8} & \frac{{A^0}^2 + {B^0}^2}{10}
\end{array}\right].$$
\justify
Using~(\ref{eq:J_part_1}) and~(\ref{eq:J_part_2}), in equation~(\ref{eq:J_taylor}), we have:
\begin{equation*}
\lim_{\textnormal{min}\{M, N\} \rightarrow \infty}(\tilde{\boldsymbol{\theta}} - \boldsymbol{\theta}^0)\mathbf{D}^{-1}  =  \lim_{\textnormal{min}\{M, N\} \rightarrow \infty}(\hat{\boldsymbol{\theta}} - \boldsymbol{\theta}^0 )\mathbf{D}^{-1}.
\end{equation*}

\justify
\vspace{-5mm}
It follows that LSE, $\hat{\boldsymbol{\theta}}$ and ALSE, $\tilde{\boldsymbol{\theta}}$ of $\boldsymbol{\theta}^0$ of model~(\ref{eq:model_1}) are asymptotically equivalent in distribution.\\ \qed
\vspace{-6mm}
\section*{Appendix C}\label{appendix:C}
\vspace{-5mm}
\noindent The following lemmas are required to prove the Theorem~\ref{theorem:3}:
\begin{lemma}\label{lemma:8}
Consider the set $S_{c}^{\boldsymbol{\vartheta}_1^0}$ defined in Lemma~\ref{lemma:6}. If for some c $>$0, $$\limsup \sup_{S_c^{\boldsymbol{\vartheta}_1^0}} \frac{1}{MN} (I_1(\boldsymbol{\vartheta}) - I_1(\boldsymbol{\vartheta}_1^0)) < 0 \ \ a.s., $$  then $\tilde{\boldsymbol{\vartheta}_1}$ is a strongly consistent estimator of  $\boldsymbol{\vartheta}_1^0$. Here $I_1(\boldsymbol{\vartheta})$ is as defined in~(\ref{eq:per_1}).
\end{lemma}
\begin{proof}
This proof can be obtained on the same lines as Lemma~\ref{lemma:6}.\\
\end{proof}
\begin{lemma}\label{lemma:9}If assumptions 1 and 3 are true, then the following holds true:
$$(\tilde{\boldsymbol{\vartheta}_1} - \boldsymbol{\vartheta}_1^0)(\sqrt{MN}\mathbf{D}_1)^{-1} \xrightarrow{a.s.} 0.$$
\end{lemma}
\begin{proof}
This proof can be obtained along the same lines as Lemma~\ref{lemma:7}  by replacing $I(\boldsymbol{\vartheta})$ by $I_1(\boldsymbol{\vartheta})$.\\
\end{proof}

\justify
\vspace{-6mm}
\emph{Proof of Theorem 3:}
First we prove the consistency of the estimates of the non-linear parameters of the first component of the model. For notational simplicity, we assume $p = 2$. 
\justify
\vspace{-5mm}
We consider the difference:  $\frac{1}{MN}\bigg(I_1(\boldsymbol{\vartheta}) - I_1(\boldsymbol{\vartheta}_1^0)\bigg) = \frac{1}{MN}\bigg(I_1(\alpha,\beta, \gamma, \delta) - I_1(\alpha_1^0, \beta_1^0, \gamma_1^0, \delta_1^0)\bigg)$
\vspace{-5mm}
\begin{flalign*}
= \frac{1}{(MN)^2}\Bigg[\Bigg|\sum_{n=1}^{N}\sum_{m=1}^{M}y(m, n)e^{-i(\alpha m + \beta m^2 + \gamma n + \delta n^2)}\Bigg| - \Bigg|\sum_{n=1}^{N}\sum_{m=1}^{M}y(m, n)e^{-i(\alpha_1^0 m + \beta_1^0 m^2 + \gamma_1^0 n + \delta_1^0 n^2)}\Bigg|\Bigg.] &
\end{flalign*}

\justify
\vspace{-5mm}
The set  $S_c^{\boldsymbol{\vartheta}_1^0}$, can be split into two parts and written as ${S_c^{\boldsymbol{\vartheta}_1^0}}^1 \cup {S_c^{\boldsymbol{\vartheta}_1^0}}^2$, where \\ \\
${S_c^{\boldsymbol{\vartheta}_1^0}}^1 = \{(\boldsymbol{\vartheta}: |\boldsymbol{\vartheta} - \boldsymbol{\vartheta}_1^0| > 4c; \ \boldsymbol{\vartheta} = \boldsymbol{\vartheta}_2^0\}$ and ${S_c^{\boldsymbol{\vartheta}_1^0}}^2 = \{(\boldsymbol{\vartheta}: |\boldsymbol{\vartheta} - \boldsymbol{\vartheta}_2^0| > 4c; \ \boldsymbol{\vartheta} \neq \boldsymbol{\vartheta}_2^0\}$.
\begin{flalign*}
&\limsup_{\textmd{min}\{M, N\} \rightarrow \infty} \sup_{S_c^1}\frac{1}{MN}\bigg(I_1(\alpha,\beta, \gamma, \delta) - I_1(\alpha_1^0,\beta_1^0, \gamma_1^0, \delta_1^0)\bigg)&\\
& = \limsup_{\textmd{min}\{M, N\} \rightarrow \infty}\sup_{S_c^1}\frac{1}{(MN)^2} \Bigg[\Bigg\{\sum_{n=1}^{N}\sum_{m=1}^{M}\bigg(\sum_{k=1}^2 A_k^0 \cos(\alpha_k^0 m + \beta_k^0 m^2 + \gamma_k^0 n + \delta_k^0 n^2) + & \\ 
& \quad B_k^0 \sin(\alpha_k^0 m + \beta_k^0 m^2 +  + \gamma_k^0 n + \delta_k^0 n^2) + X(m ,n)\bigg)\cos(\alpha m + \beta m^2 + \gamma n + \delta n^2)\Bigg\}^2\Bigg] + &\\
& \quad \limsup_{\textmd{min}\{M, N\} \rightarrow \infty}\sup_{S_c^1} \frac{1}{(MN)^2} \Bigg[\Bigg\{\sum_{n=1}^{N}\sum_{m=1}^{M}\bigg(\sum_{k=1}^2 A_k^0 \cos(\alpha_k^0 m + \beta_k^0 m^2 + \gamma_k^0 n + \delta_k^0 n^2) + & \\
& \quad B_k^0 \sin(\alpha_k^0 m + \beta_k^0 m^2 +  + \gamma_k^0 n + \delta_k^0 n^2) + X(m, n)\bigg) \sin(\alpha m + \beta m^2 + \gamma n + \delta n^2)\Bigg\}^2\Bigg] - &\\
& \quad \limsup_{\textmd{min}\{M, N\} \rightarrow \infty}\sup_{S_c^1} \frac{1}{(MN)^2} \Bigg[\Bigg\{\sum_{n=1}^{N}\sum_{m=1}^{M}\bigg(\sum_{k=1}^2 A_k^0 \cos(\alpha_k^0 m + \beta_k^0 m^2 + \gamma_k^0 n + \delta_k^0 n^2) + & \\
& \quad B_k^0 \sin(\alpha_k^0 m + \beta_k^0 m^2 + \gamma_k^0 n + \delta_k^0 n^2) + X(m, n)\bigg) \cos(\alpha_1^0 m + \beta_1^0 m^2 +  + \gamma_1^0 n + \delta_1^0 n^2)\Bigg\}^2\Bigg] - &\\
& \quad \limsup_{\textmd{min}\{M, N\}  \rightarrow \infty}\sup_{S_c^1}\frac{1}{(MN)^2} \Bigg[\Bigg\{\sum_{n=1}^{N}\sum_{m=1}^{M}\bigg(\sum_{k=1}^2 A_k^0 \cos(\alpha_k^0 m + \beta_k^0 m^2 +  + \gamma_k^0 n + \delta_k^0 n^2) + & \\
& \quad B_k^0 \sin(\alpha_k^0 m + \beta_k^0 m^2 +  + \gamma_k^0 n + \delta_k^0 n^2) + X(m, n)\bigg) \sin(\alpha_1^0 m + \beta_1^0 m^2 + \gamma_1^0 n + \delta_1^0 n^2)\Bigg\}^2\Bigg]&\\
& = \frac{1}{4}({A_2^0}^2 + {B_2^0}^2 -{A_1^0}^2 - {B_1^0}^2) < 0 \ \ a.s.  \ (\textmd{Assumption} \ 4.)&.\\ 
&\textmd{Similarly,}& \\
& \limsup_{\textmd{min}\{M, N\} \rightarrow \infty}\sup_{S_c^2}\frac{1}{MN}\bigg(I_1(\alpha, \beta, \gamma, \delta) - I_1(\alpha_1^0,\beta_1^0, \gamma_1^0, \delta_1^0)\bigg) =  \frac{1}{4}(0 + 0 - {A_1^0}^2 - {B_1^0}^2) < 0 \ \ a.s. &
\end{flalign*}
Therefore, using Lemma~\ref{lemma:8}, we get the consistency of the non-linear parameters of the first component. Now to prove the consistency of linear parameter estimators $\tilde{A_1}$ and $\tilde{B_1}$, observe that
\begin{equation*}\begin{split}
\tilde{A_1} & = \frac{2}{MN}\sum_{n=1}^{N}\sum_{m = 1}^{M} y(m, n)\cos(\tilde{\alpha_1} m + \tilde{\beta_1} m^2 + \tilde{\gamma_1} n + \tilde{\delta_1} n^2) \\
& = \frac{2}{MN}\sum_{n=1}^{N}\sum_{m=1}^{M} \bigg(\sum_{k=1}^{p}A_k^0\cos(\alpha_k^0 m + \beta_k^0 m^2 + \gamma_k^0 n + \delta_k^0 n^2) + B_k^0\sin(\alpha_k^0 m + \beta_k^0 m^2 + \gamma_k^0 n + \delta_k^0 n^2)\\
& \quad + X(m, n)\bigg)\cos(\tilde{\alpha_1} m + \tilde{\beta_1} m^2 + \tilde{\gamma_1} n + \tilde{\delta_1} n^2).
\end{split}\end{equation*}
\setlength{\belowdisplayskip}{0pt} \setlength{\belowdisplayshortskip}{0pt}
\setlength{\abovedisplayskip}{0pt} \setlength{\abovedisplayshortskip}{0pt}
\justify
\vspace{-10mm}
Using Lemma~\ref{lemma:1}, part \textit{(g)} $\frac{2}{MN}\sum\limits_{n=1}^{N}\sum\limits_{m=1}^M X(m, n)\cos(\tilde{\alpha_1} m + \tilde{\beta_1} m^2 + \tilde{\gamma_1} n + \tilde{\delta_1} n^2 ) \rightarrow 0$. Now expanding $\cos(\tilde{\alpha_1} m + \tilde{\beta_1} m^2 + \tilde{\gamma_1} n + \tilde{\delta_1} n^2)$ by multivariate Taylor series around $(\alpha^0 , \beta^0, \gamma^0, \delta^0)$ and using Lemmas~\ref{lemma:9} and~\ref{lemma:1},\textit{(c)-(f)}, we get the desired result. \\ \qed
\justify
\vspace{-5mm}
\emph{Proof of Theorem 4:}
From Theorem~\ref{theorem:3} and Lemmas~\ref{lemma:7} and~\ref{lemma:8}, we have the following: \\
\begin{equation*}\begin{split}
\tilde{A_1} & = A_1^0 + o(1) \quad  \hspace{5mm}\tilde{B_1} = B_1^0 + o(1), \\
\tilde{\alpha_1} & = \alpha_1^0 + o(M^{-1}) \quad \tilde{\beta_1} = \beta_1^0 + o(M^{-2}).\\
\tilde{\gamma_1} & = \gamma_1^0 + o(N^{-1}) \quad \hspace{1mm} \tilde{\delta_1} = \delta_1^0 + o(N^{-2}).\\
\end{split}\end{equation*}
Thus,\\
\begin{equation}\begin{split}\label{eq:relation_est_true}
&\tilde{A_1} \cos(\tilde{\alpha_1} m + \tilde{\beta_1} m^2 +\tilde{\gamma_1} n + \tilde{\delta_1} n^2) + \tilde{B_1} \sin(\tilde{\alpha_1} m + \tilde{\beta_1} m^2 + \tilde{\gamma_1} n + \tilde{\delta_1} n^2) \\
& \quad \quad \quad \quad  = A_1^0 \cos(\alpha_1^0 m + \beta_1^0 m^2 + \gamma_1^0 n + \delta_1^0 n^2) + B_1^0  \sin(\alpha_1^0 m + \beta_1^0 m^2 + \gamma_1^0 n + \delta_1^0 n^2) + o(1).\\
\end{split}\end{equation}
\justify
\vspace{-3mm}
Consider the set  $S_c^{\boldsymbol{\vartheta}_2^0} $, that can be split into $p$ sets and written as ${S_c^{\boldsymbol{\vartheta}_2^0}}^1 \cup {S_c^{\boldsymbol{\vartheta}_2^0}}^2 \cup \cdots \cup {S_c^{\boldsymbol{\vartheta}_2^0}}^p$, where 
\begin{flalign*}
&{S_c^{\boldsymbol{\vartheta}_2^0}}^1 = \{\boldsymbol{\vartheta}: |\boldsymbol{\vartheta} - \boldsymbol{\vartheta}_2^0| > 4c; \ \boldsymbol{\vartheta} = \boldsymbol{\vartheta}_1^0\},& \\
&{S_c^{\boldsymbol{\vartheta}_2^0}}^2 = \{\boldsymbol{\vartheta}: |\boldsymbol{\vartheta} - \boldsymbol{\vartheta}_2^0| > 4c; \ \boldsymbol{\vartheta} = \boldsymbol{\vartheta}_3^0\},& \\
&\vdots&\\
&{S_c^{\boldsymbol{\vartheta}_2^0}}^{p-1} = \{\boldsymbol{\vartheta}: |\boldsymbol{\vartheta} - \boldsymbol{\vartheta}_2^0| > 4c; \ \boldsymbol{\vartheta} = \boldsymbol{\vartheta}_p^0\}, \textnormal{ and} & \\
&{S_c^{\boldsymbol{\vartheta}_2^0}}^p = \{\boldsymbol{\vartheta}: |\boldsymbol{\vartheta} - \boldsymbol{\vartheta}_2^0| > 4c; \ \boldsymbol{\vartheta} \neq \boldsymbol{\vartheta}_k^0,\textmd{ for any }k = 1, \cdots, p\}.&
\end{flalign*}
\justify
\vspace{-3mm}
Now we consider the following difference: 
\begin{equation}\begin{split}\label{eq:per_2}
\frac{1}{MN}\bigg(I_2(\alpha, \beta, \gamma, \delta) - I_2(\alpha_2^0, \beta_2^0, \gamma_2^0, \delta_2^0)\bigg) & = \frac{1}{(MN)^2}\Bigg|\sum_{n=1}^{N}\sum_{m=1}^{M}y^1(m, n)e^{-i(\alpha m + \beta m^2 + \gamma n + \delta n^2)}\Bigg|^2 - \\
& \quad \frac{1}{(MN)^2} \Bigg|\sum_{n=1}^{N}\sum_{m=1}^{M}y^1(m, n)e^{-i(\alpha_2^0 m + \beta_2^0 m^2 + \gamma_2^0 n + \delta_2^0 n^2)}\Bigg|^2. \\
\end{split}\end{equation}
Here, $y^1(m, n) = y(m, n) - \tilde{A_1} \cos(\tilde{\alpha_1} m + \tilde{\beta_1} m^2 + \tilde{\gamma_1} n + \tilde{\delta_1} n^2) + \tilde{B_1} \sin(\tilde{\alpha_1} m + \tilde{\beta_1} m^2 + \tilde{\gamma_1} n + \tilde{\delta_1} n^2)$, that is the new data obtained by removing the effect of the first component from the observed data $y(m, n)$. Using~(\ref{eq:relation_est_true}), we have $$y^1(m, n) = o(1) + \sum_{k=2}^p (A_k^0 \cos(\alpha_k^0 m + \beta_k^0 m^2 + \gamma_k^0 n + \delta_k^0 n^2) + B_k^0 \sin(\alpha_k^0 m + \beta_k^0 m^2 + \gamma_k^0 n + \delta_k^0 n^2)) + X(m, n).$$\\
Substituting this in~(\ref{eq:per_2}), it can be easily seen that:
\begin{equation*}
\limsup_{\textmd{min}\{M, N\} \rightarrow \infty}\sup_{S_c^k}\frac{1}{MN}\bigg(I_2(\alpha, \beta, \gamma, \delta) - I_2(\alpha_2^0, \beta_2^0, \gamma_2^0, \delta_2^0)\bigg) < 0 \ \  a.s.  \ \forall \ k = 1, \cdots, p.\\
\end{equation*}
Combining, we have $\limsup\limits_{\textmd{min}\{M, N\} \rightarrow \infty}\sup\limits_{S_c}\frac{1}{MN}\bigg(I_2(\alpha, \beta, \gamma, \delta) - I_2(\alpha_2^0, \beta_2^0, \gamma_2^0, \delta_2^0) \bigg)$ $<$ 0 $a.s.$ \\
Therefore, $\tilde{\alpha_2}$ $\xrightarrow{a.s.}$ $\alpha_2^0$, $\tilde{\beta_2}$ $\xrightarrow{a.s.}$ $\beta_2^0$, $\tilde{\gamma_2}$ $\xrightarrow{a.s.}$ $\gamma_2^0$ and $\tilde{\delta_2}$ $\xrightarrow{a.s.}$ $\delta_2^0$ by Lemma~\ref{lemma:8}. Following the same argument as in Theorem~\ref{theorem:3}, we can prove the consistency of linear parameter estimators $\tilde{A_2}$ and $\tilde{B_2}$.
\justify
\vspace{-5mm}
Proceeding in a similar way, it can be shown that $\tilde{\boldsymbol{\theta}_k} \xrightarrow{a.s.} \boldsymbol{\theta}_k^0$ for $3 \leqslant k \leqslant p$.\\ \qed

\justify
\vspace{-5mm}
\emph{Proof of Theorem 6:} The ALSEs of the linear parameters $A$ and $B$ are given by:
\begin{equation*}\begin{split}
\tilde{A} & = \frac{2}{MN} \sum_{n=1}^{N}\sum_{m=1}^{M} y(m, n) \cos(\tilde{\alpha} m + \tilde{\beta} m^2 + \tilde{\gamma} n + \tilde{\delta} n^2), \textmd{ and } \\
\tilde{B} & = \frac{2}{MN} \sum_{n=1}^{N}\sum_{m=1}^{M} y(m, n) \sin(\tilde{\alpha} m + \tilde{\beta} m^2 +  \tilde{\gamma} n + \tilde{\delta} n^2).
\end{split}\end{equation*}
For $k = p+1$,
\begin{equation*} 
 \tilde{A}_{p+1} = \frac{2}{n} \sum\limits_{n=1}^{N}\sum\limits_{m=1}^{M} y^{p+1}(m, n) \cos(\tilde{\alpha}_{p+1} m + \tilde{\beta}_{p+1}m^2 + \tilde{\gamma}_{p+1} n + \tilde{\delta}_{p+1} n^2)
\end{equation*}
\justify
\vspace{-5mm}
where $\tilde{\alpha}_{p+1}$, $\tilde{\beta}_{p+1}$, $\tilde{\gamma}_{p+1}$ and $\tilde{\delta}_{p+1}$ are obtained by maximising $I_{p+1}(\alpha,\beta,\gamma,\delta)$, and \\
$y^{p+1}(m, n) = y(m, n) - \bigg(\sum\limits_{k=1}^{p} \tilde{A}_k \cos(\tilde{\alpha}_{k} m + \tilde{\beta}_{k} m^2 + \tilde{\gamma}_{k} n + \tilde{\delta}_{k} n^2) + \tilde{B}_k \cos(\tilde{\alpha}_{k} m + \tilde{\beta}_{k}m^2  + \tilde{\gamma}_{k} n + \tilde{\delta}_{k} n^2)\bigg)$. Using~(\ref{eq:relation_est_true}), we have:
\begin{equation*}\begin{split}
\tilde{A}_{p+1} & = \frac{2}{MN}\sum_{n=1}^{N}\sum_{m=1}^{M} X(m, n) \cos(\tilde{\alpha}_{p+1} m + \tilde{\beta}_{p+1} m^2 + \tilde{\gamma}_{p+1} n + \tilde{\delta}_{p+1} n^2 ) + o(1),
\end{split}\end{equation*}
\begin{equation*}\begin{split}
\tilde{B}_{p+1} & = \frac{2}{MN}\sum_{n=1}^{N}\sum_{m=1}^{M}X(m, n) \sin(\tilde{\alpha}_{p+1} m + \tilde{\beta}_{p+1} m^2 + \tilde{\gamma}_{p+1} n + \tilde{\delta}_{p+1} n^2 ) + o(1).       
\end{split}\end{equation*}
Therefore, using Lemma~\ref{lemma:1}, part \textit{(g)}, we have $\tilde{A}_{p+1} \xrightarrow{a.s.} 0 $ and $\tilde{B}_{p+1} \xrightarrow{a.s.} 0. $ \\ \qed 
\vspace{-10mm}
\section*{Appendix D}\label{appendix:D}
\textit{Proof of Lemma 3:} We need to compute:
\begin{flalign*}
& \lim_{\textmd{min}\{M,N\} \rightarrow \infty}\frac{1}{M^s N^t\sqrt{MN}}\sum_{n=1}^{N}\sum_{m=1}^M m^s n^t \cos(\omega_1 m + \omega_2 m^2   + \omega_3 n + \omega_4 n^2)\cos(\psi_1 m + \psi_2 m^2   + \psi_3 n + \psi_4 n^2)& \\
& = \lim\limits_{\textmd{min}\{M,N\} \rightarrow \infty}\frac{1}{2M^s N^t\sqrt{MN}}\sum\limits_{n=1}^{N}\sum\limits_{m=1}^M m^s n^t \cos((\omega_1 + \psi_1) m + (\omega_2 + \psi_2) m^2   + (\omega_3 + \psi_3) n + (\omega_4 + \psi_4) n^2) + & \\
& \quad \lim\limits_{\textmd{min}\{M,N\} \rightarrow \infty}\frac{1}{2M^s N^t\sqrt{MN}}\sum\limits_{n=1}^{N}\sum\limits_{m=1}^M m^s n^t \cos((\omega_1 - \psi_1) m + (\omega_2 - \psi_2) m^2   + (\omega_3 - \psi_3) n + (\omega_4 - \psi_4) n^2).&\\
&\textmd{Now we consider the exponential sum:}& \\
&\sum\limits_{n=1}^{N}\sum\limits_{m=1}^M m^s n^t \exp(i(\omega_1 + \psi_1) m + (\omega_2 + \psi_2) m^2   + (\omega_3 + \psi_3) n + (\omega_4 + \psi_4) n^2)&\\
& = \bigg(\sum\limits_{m=1}^M m^s \exp(i(\omega_1 + \psi_1) m + (\omega_2 + \psi_2) m^2)\bigg) \bigg(\sum\limits_{n=1}^{N} n^t \exp(i(\omega_3 + \psi_3) n + (\omega_4 + \psi_4) n^2)\bigg)& \\
& = o(M^s \sqrt{M}). o(N^t \sqrt{N}) = o(M^s N^t \sqrt{MN})\textmd{, using Lemma~\ref{lemma:4}.}&
\vspace{5mm}
\end{flalign*}
This proves \textit{(a)}. The proof of \textit{(b)} and \textit{(c)} follows along the same lines.\\ \qed
\justify
\vspace{-10mm}
\begin{flalign*}
&  \textmd{\emph{Proof of Theorem 6:} Consider the following sum of squares:}& \\
& Q_1(\boldsymbol{\theta}) = \sum_{n = 1}^{N}\sum_{m = 1}^{M}\bigg(y(m, n) - A\cos(\alpha m + \beta m^2 + \gamma n + \delta n^2) - B \sin(\alpha m + \beta m^2 + \gamma n + \delta n^2)\bigg)^2 &\\
& = \sum_{n = 1}^{N}\sum_{m = 1}^{M} y(m, n)^2 + \sum_{n = 1}^{N}\sum_{m = 1}^{M} \bigg(A\cos(\alpha m + \beta m^2 + \gamma n + \delta n^2) + B \sin(\alpha m + \beta m^2 + \gamma n + \delta n^2)\bigg)^2 & \\
& \quad \quad  - 2\sum_{n = 1}^{N}\sum_{m = 1}^{M} y(m, n)\bigg(A\cos(\alpha m + \beta m^2 + \gamma n + \delta n^2) + B \sin(\alpha m + \beta m^2 + \gamma n + \delta n^2)\bigg) &\\
& = C_1 - J_1(\boldsymbol{\theta}) + o(1), \textmd{ using Lemma~\ref{lemma:1}}, parts \textit{(c) - (f)}. &
\end{flalign*}
\begin{flalign*}
\textmd{Here, } C_1 & = \sum_{n = 1}^{N}\sum_{m = 1}^{M} y(m, n)^2,\textmd{ and} & \\
J_1(\boldsymbol{\theta}) & = 2\sum_{n = 1}^{N}\sum_{m = 1}^{M} y(m, n)\bigg(A\cos(\alpha m + \beta m^2 + \gamma n + \delta n^2) + & \\ 
& \hspace{38mm} B \sin(\alpha m + \beta m^2 + \gamma n + \delta n^2)\bigg)- MN\frac{A^2 + B^2}{2}&
\end{flalign*}
\noindent Note that at $\tilde{A}$ and $\tilde{B}$,
$J_1(\tilde{A}, \tilde{B}, \alpha, \beta, \gamma, \delta) = I_1(\alpha , \beta, \gamma, \delta)$.
Hence the estimator of $\boldsymbol{\theta}_1^0$ which maximizes $J_1(\boldsymbol{\theta})$ is equivalent to $\tilde{\boldsymbol{\theta}_1}$, the ALSE of $\boldsymbol{\theta}_1^0$. Thus, $\mathbf{J}_1'(\tilde{\boldsymbol{\theta}_1}) = 0$. Now by Taylor series expansion:
\begin{equation}\label{eq:J_1_taylor}
(\tilde{\boldsymbol{\theta}_1} - \boldsymbol{\theta}_1^0)\mathbf{D}^{-1} = -[\mathbf{J}_1'(\boldsymbol{\theta}_1^0)\mathbf{D}][\mathbf{D}\mathbf{J}_1''(\bar{\boldsymbol{\theta}_1})\mathbf{D}]^{-1}
\end{equation}
\noindent Now consider the first derivative vector, 
\begin{equation*}
\mathbf{J}_1'(\boldsymbol{\theta}_1^0) = \begin{pmatrix}
\frac{\partial J_1(\boldsymbol{\theta}_1^0)}{\partial A} & \frac{\partial J_1(\boldsymbol{\theta}_1^0)}{\partial B} & \frac{\partial J_1(\boldsymbol{\theta}_1^0)}{\partial \alpha} & \frac{\partial J_1(\boldsymbol{\theta}_1^0)}{\partial \beta} & \frac{\partial J_1(\boldsymbol{\theta}_1^0)}{\partial \gamma} & \frac{\partial J_1(\boldsymbol{\theta}_1^0)}{\partial \delta} 
\end{pmatrix}
\end{equation*}
Computing the elements of this vector, we get:\\
\begin{flalign*}
& \frac{\partial J_1(\boldsymbol{\theta}_1^0)}{\partial A} = 2 \sum_{n=1}^{N}\sum_{m=1}^{M} y(m, n) \cos(\alpha_1^0 m + \beta_1^0 m^2 + \gamma_1^0 n + \delta_1^0 n^2) - MN A_1^0 &\\
& \frac{\partial J_1(\boldsymbol{\theta}_1^0)}{\partial B}  = 2 \sum_{n=1}^{N}\sum_{m=1}^{M} y(m, n) \sin(\alpha_1^0 m + \beta_1^0 m^2 + \gamma_1^0 n + \delta_1^0 n^2) - MN B_1^0 & \\
& \frac{\partial J_1(\boldsymbol{\theta}_1^0)}{\partial \alpha}  = 2 \sum_{n=1}^{N}\sum_{m=1}^{M} m\ y(m, n) \xi(m, n) \quad  \quad \frac{\partial J_1(\boldsymbol{\theta}_1^0)}{\partial \beta}  = 2 \sum_{n=1}^{N}\sum_{m=1}^{M} m^2 \ y(m, n) \xi(m, n) & \\
& \frac{\partial J_1(\boldsymbol{\theta}_1^0)}{\partial \gamma}  = 2 \sum_{n=1}^{N}\sum_{m=1}^{M} n\ y(m, n) \xi(m, n) \quad \quad \frac{\partial J_1(\boldsymbol{\theta}_1^0)}{\partial \delta}  = 2 \sum_{n=1}^{N}\sum_{m=1}^{M} n^2\ y(m, n)\xi(m, n) &\\
& \textmd{where, }\xi(m, n) = \bigg(-A_1^0 \sin(\alpha_1^0 m + \beta_1^0 m^2 + \gamma_1^0 n + \delta_1^0 n^2) + B_1^0 \cos(\alpha_1^0 m + \beta_1^0 m^2 + \gamma_1^0 n + \delta_1^0 n^2)\bigg)&
\end{flalign*}
\noindent Using the results in Lemma~\ref{lemma:5}, we see that $\mathbf{J}_1'(\boldsymbol{\theta}_1^0)\mathbf{D}$  is asymptotically equivalent to:\\
\begin{flalign*}
\begin{pmatrix}
&\frac{2}{\sqrt{MN}} \sum\limits_{n=1}^N\sum\limits_{m=1}^M X(m, n)\cos(\alpha_1^0 m + \beta_1^0 m^2 + \gamma_1^0 n + \delta_1^0 n^2) & \\
& \frac{2}{\sqrt{MN}} \sum\limits_{n=1}^N\sum\limits_{m=1}^M X(m, n)\sin(\alpha_1^0 m + \beta_1^0 m^2 + \gamma_1^0 n + \delta_1^0 n^2)& \\
&\frac{2}{M^{3/2}N^{1/2}} \sum\limits_{n=1}^N \sum\limits_{m=1}^M m\ X(m, n)\xi(m, n)& \\
&\frac{2}{M^{5/2}N^{1/2}} \sum\limits_{n=1}^N \sum\limits_{m=1}^M m^2\ X(m, n)\xi(m, n)& \\
& \frac{2}{M^{1/2}N^{3/2}} \sum\limits_{n=1}^N \sum\limits_{m=1}^M n\ X(m, n)\xi(m, n) &\\
&\frac{2}{M^{1/2}N^{5/2}} \sum\limits_{n=1}^N \sum\limits_{m=1}^M n^2\ X(m, n)\xi(m, n) &
\end{pmatrix}.
\end{flalign*}
\justify Now using Central Limit Theorem of the stochastic processes, see Fuller \cite{2009}, the $6 \times 1$ vector in the right hand side of the above equation tends to a normal vector with mean 0 and variance-covariance matrix $4c\sigma^2\boldsymbol{\Sigma}_1$. Thus, we have \\
\begin{equation}\label{eq:J_1_part_1}
\mathbf{J}_1'(\boldsymbol{\theta}_1^0)\mathbf{D} \xrightarrow{d} N_6(0,4c\sigma^2\boldsymbol{\Sigma}_1)
\end{equation}
\justify
Here, $$\boldsymbol{\Sigma}_1 = \begin{pmatrix}
\frac{1}{2} & 0 & \frac{B_1^0}{4} & \frac{B_1^0}{6} & \frac{B_1^0}{4} & \frac{B_1^0}{6} \\
0 & \frac{1}{2} & \frac{-A_1^0}{4} & \frac{-A_1^0}{6} & \frac{-A_1^0}{4} & \frac{-A_1^0}{6}\\
\frac{B_1^0}{4} & \frac{-A_1^0}{4} & \frac{{A_1^0}^2 + {B_1^0}^2}{6} & \frac{{A_1^0}^2 + {B_1^0}^2}{8} & \frac{{A_1^0}^2 + {B_1^0}^2}{8} & \frac{{A_1^0}^2 + {B_1^0}^2}{12} \\
\frac{B_1^0}{6} & \frac{-A_1^0}{6} & \frac{{A_1^0}^2 + {B_1^0}^2}{8} & \frac{{A_1^0}^2 + {B_1^0}^2}{10} & \frac{{A_1^0}^2 + {B_1^0}^2}{12} & \frac{{A_1^0}^2 + {B_1^0}^2}{18} \\
\frac{B_1^0}{4} & \frac{-A_1^0}{4} & \frac{{A_1^0}^2 + {B_1^0}^2}{8} & \frac{{A_1^0}^2 + {B_1^0}^2}{12} & \frac{{A_1^0}^2 + {B_1^0}^2}{6} & \frac{{A_1^0}^2 + {B_1^0}^2}{8} \\
\frac{B_1^0}{6} & \frac{-A_1^0}{6} & \frac{{A_1^0}^2 + {B_1^0}^2}{12} & \frac{{A_1^0}^2 + {B_1^0}^2}{18} & \frac{{A_1^0}^2 + {B_1^0}^2}{8} & \frac{{A_1^0}^2 + {B_1^0}^2}{10} \\
\end{pmatrix}.$$\\
Note that, 
\begin{equation*}
\lim\limits_{\textnormal{min}\{M, N\} \rightarrow \infty}[\mathbf{D}\mathbf{J}_1''(\bar{\boldsymbol{\theta}_1})\mathbf{D}] = \lim\limits_{\textnormal{min}\{M, N\} \rightarrow \infty}[\mathbf{D}\mathbf{J}_1''(\boldsymbol{\theta}_1^0)\mathbf{D}].
\end{equation*}
Using Lemma~\ref{lemma:1}, parts \textit{(c) - (g)} and by simple calculations computing each element of the matrix $\mathbf{D}\mathbf{J}_1''(\boldsymbol{\theta}_1^0)\mathbf{D}$, we can show:
\begin{equation}\label{eq:J_1_part_2}
\lim\limits_{\textnormal{min}\{M, N\} \rightarrow \infty}[\mathbf{D}\mathbf{J}_1''(\boldsymbol{\theta}_1^0)\mathbf{D}] = -2\boldsymbol{\Sigma}_1. \\
\vspace{5mm}
\end{equation}
Hence, using equation~(\ref{eq:J_1_part_1}) and~(\ref{eq:J_1_part_2}) in~(\ref{eq:J_1_taylor}), we have:
\vspace{2mm}
\begin{equation*}
(\tilde{\boldsymbol{\theta}_1} - \boldsymbol{\theta}_1^0) \mathbf{D}^{-1} \xrightarrow{d} N(0,c_1\sigma^2\boldsymbol{\Sigma}_1^{-1}).\\
\end{equation*}
Hence the result. \qed

\end{document}